\documentclass[11pt]{article}
\usepackage{verbatim,amsmath,eufrak}
\usepackage{hyperref}
\usepackage{cite} 
\usepackage{makeidx}

\hypersetup{
    bookmarks=true,         
    unicode=false,          
    pdftoolbar=true,        
    pdfmenubar=true,        
    pdffitwindow=false,     
    pdfstartview={FitH},    
    pdftitle={My title},    
    pdfauthor={Author},     
    pdfsubject={Subject},   
    pdfcreator={Creator},   
    pdfproducer={Producer}, 
    pdfkeywords={keywords}, 
    pdfnewwindow=true,      
    colorlinks=false,       
    linkcolor=red,          
    citecolor=green,        
    filecolor=magenta,      
    urlcolor=cyan           
}

\renewcommand{\arraystretch}{1.2}
\makeatletter
\newdimen\normalarrayskip              
\newdimen\minarrayskip                 
\normalarrayskip\baselineskip
\minarrayskip\jot
\newif\ifold             \oldtrue            \def\new{\oldfalse}
\def\arraymode{\ifold\relax\else\displaystyle\fi} 
\def\eqnumphantom{\phantom{(\theequation)}}     
\def\@arrayskip{\ifold\baselineskip\z@\lineskip\z@
     \else
     \baselineskip\minarrayskip\lineskip2\minarrayskip\fi}
\def\@arrayclassz{\ifcase \@lastchclass \@acolampacol \or
\@ampacol \or \or \or \@addamp \or
   \@acolampacol \or \@firstampfalse \@acol \fi
\edef\@preamble{\@preamble
  \ifcase \@chnum
     \hfil$\relax\arraymode\@sharp$\hfil
     \or $\relax\arraymode\@sharp$\hfil
     \or \hfil$\relax\arraymode\@sharp$\fi}}
\def\@array[#1]#2{\setbox\@arstrutbox=\hbox{\vrule
     height\arraystretch \ht\strutbox
     depth\arraystretch \dp\strutbox
     width\z@}\@mkpream{#2}\edef\@preamble{\halign
\noexpand\@halignto
\bgroup \tabskip\z@ \@arstrut \@preamble \tabskip\z@ \cr}%
\let\@startpbox\@@startpbox \let\@endpbox\@@endpbox
  \if #1t\vtop \else \if#1b\vbox \else \vcenter \fi\fi
  \bgroup \let\par\relax
  \let\@sharp##\let\protect\relax
  \@arrayskip\@preamble}
\def\eqnarray{\stepcounter{equation}%
              \let\@currentlabel=\theequation
              \global\@eqnswtrue
              \global\@eqcnt\z@
              \tabskip\@centering
              \let\\=\@eqncr
 \halign to \displaywidth\bgroup
    \eqnumphantom\@eqnsel\hskip\@centering
    $\displaystyle \tabskip\z@ {##}$%
    \global\@eqcnt\@ne \hskip 2\arraycolsep
         $\displaystyle\arraymode{##}$\hfil
    \global\@eqcnt\tw@ \hskip 2\arraycolsep
         $\displaystyle\tabskip\z@{##}$\hfil
         \tabskip\@centering
    &{##}\tabskip\z@\cr}
\begingroup\ifx\undefined\newsymbol \else\def\input#1 {\endgroup}\fi

\newcounter{app}

\def\app{\setcounter{equation}{0}
\def\theequation{A\Roman{app}.\arabic{equation}}\par
   \addvspace{4ex}
   \@afterindentfalse
  \secdef\@app\@dapp}
\newcommand\@app{\@startsection {app}{1}{0ex}%
                                   {-3.5ex \@plus -1ex \@minus -.2ex}%
                                   {2.3ex \@plus.2ex}%
                                   {\normalfont\Large\bf}}

\def\@dapp#1{%
{\parindent \z@ \raggedright  \bf #1}\par\nobreak}
\def\l@app#1#2{\ifnum \c@tocdepth >\z@
    \addpenalty\@secpenalty
    \addvspace{1.0em \@plus\p@}%
    \setlength\@tempdima{8.5em}%
    \begingroup
      \parindent \z@ \rightskip \@pnumwidth
      \parfillskip -\@pnumwidth
      \leavevmode \bfseries
      \advance\leftskip\@tempdima
      \hskip -\leftskip
      #1\nobreak\hfil \nobreak\hb@xt@\@pnumwidth{\hss #2}\par
    \endgroup\fi}
\newcounter{sapp}[app]

\def\sapp{\def\theequation{A\arabic{app}.\arabic{equation}}\par
   \@afterindentfalse
  \secdef\@sapp\@dsapp}
\newcommand\@sapp{\@startsection{sapp}{2}{\z@}%
                                     {-3.25ex\@plus -1ex \@minus -.2ex}%
                                     {1.5ex \@plus .2ex}%
                                     {\normalfont\large\bfseries}}

\def\@dsapp#1{%
{\parindent \z@ \raggedright  \bf #1}\par\nobreak}
\newcommand{\l@sapp}{\@dottedtocline{2}{1.5em}{3em}}
\def\draft{\oddsidemargin -.5truein
        \def\@oddfoot{\sl preliminary draft \hfil
        \rm\thepage\hfil\sl\today\quad\militarytime}
        \let\@evenfoot\@oddfoot \overfullrule 3pt
        \let\label=\draftlabel
        \let\marginnote=\draftmarginnote
   \def\@eqnnum{(\theequation)\rlap{\kern\marginparsep\tt\@eqnlabel}%
\global\let\@eqnlabel\@vacuum}  }

\def\numberbysection{\@addtoreset{equation}{section}
        \def\theequation{\thesection.\arabic{equation}}}

\numberbysection
\newtheorem{conj}{Conjecture}[section]

\def\be{\begin{eqnarray}}
\def\ee{\end{eqnarray}}

\def\p{\partial}

\def\beq{\begin{equation}}
\def\eeq{\end{equation}}
\def\ba{\beq\new\begin{array}{c}}
\def\ea{\end{array}\eeq}
\def\be{\ba}
\def\ee{\ea}

\def\tr{{\rm tr}\,}
\def\Tr{{\rm Tr}\,}
\def\dim{{\rm dim}\,}
\def\diag{{\rm diag}\,}
\def\psistar{\psi^{*}}
\newfont{\Bbbb}{msbm7 scaled 1\@ptsize00}

\newcommand{\z}{\raise-1pt\hbox{$\mbox{\Bbbb Z}$}}
\def\normordboson{ {\scriptstyle {{*}\atop{*}}} }
\def\lbr{\left <}
\def\rbr{\right >}

\def\lvac{\left <0\right |}
\def\rvac{\left |0\right >}
\def\lvacn{\left <n\right |}
\def\rvacn{\left |n\right >}

\def\normord{ {\scriptstyle {{\bullet}\atop{\bullet}}} }

\newfont{\alef}{msbm10 at 11pt}
\newfont {\goth}{eufm10 at 11pt}
\def\mathbb#1{\hbox{{\alef #1}}}

\let\@@savethanks\thanks
\def\thanks#1{\gdef\thefootnote{\alph{footnote}}\@@savethanks{#1}}

\newtheorem{theorem}{Theorem}[section]
\newtheorem{lemma}[theorem]{Lemma}

\newenvironment{proof}[1][Proof.]{\begin{trivlist}
\item[\hskip \labelsep {\bfseries #1}]}{\end{trivlist}}

\makeatletter
\g@addto@macro \normalsize {%
 \setlength\abovedisplayskip{14pt plus 3pt minus 3pt}%
 \setlength\belowdisplayskip{14pt plus 3pt minus 3pt}%
  \setlength\abovedisplayshortskip{11pt plus 3pt minus 3pt}%
 \setlength\belowdisplayshortskip{11pt plus 3pt minus 3pt}%
}
\makeatother

\bigskip
\bigskip

\title{
\bigskip
{\bf
Enumerative geometry, tau-functions and Heisenberg--Virasoro algebra} \vspace{.5cm}}
\author{{\bf A. Alexandrov}\thanks{E-mail:  {\tt alexandrovsash at gmail.com}}
\date{ } \\
{\small {\it Freiburg Institute for Advanced Studies (FRIAS), University of Freiburg, Germany \&}}\\
{\small {\it  Mathematics Institute, University of Freiburg, Germany \&}}\\
{\small {\it ITEP, Moscow, Russia}}\\
}

\begin{document}

\setcounter{footnote}{0}

\setcounter{tocdepth}{3}

\maketitle

\vspace{-8.0cm}

\begin{center}
\hfill ITEP/TH-03/14
\end{center}

\vspace{6.5cm}
\begin{abstract} 
In this paper we establish relations between three enumerative geometry tau-functions, namely the Kontsevich--Witten, Hurwitz and Hodge tau-functions. The relations allow us to describe the tau-functions in terms of matrix integrals, Virasoro constraints and Kac--Schwarz operators. All constructed operators belong to the algebra (or group) of symmetries of the KP hierarchy.
\end{abstract}

\bigskip


{\small \bf MSC 2010 Primary: 37K10, 81R12, 81R10, 14N10; Secondary: 81T30.}

\newpage 

\tableofcontents

\def\thefootnote{\arabic{footnote}}

\section*{Introduction}
\addcontentsline{toc}{section}{Introduction}
\setcounter{equation}{0}

The generating functions of enumerative geometry constitute an important and very interesting class of tau-functions of integrable hierarchies. This intriguing paradigm remains a subject of considerable interest and a number of examples of its manifestation continuously grows. In this paper we construct new relations connecting three tau-functions from this family and derive linear constraints for all of them.

Namely, we consider the Kontsevich--Witten tau-function, the generating function of linear Hodge integrals and the generating function of the simple Hurwitz numbers. The first of them is known to be a tau-function of the KdV integrable hierarchy, while the other two are tau-functions of the KP hierarchy. In this work we prove that
\be\label{intM}
\tau_{KW}=\widehat{G}_{+}\, \tau_{Hodge}.
\ee
Here the operator $\widehat{G}_{+}$ belongs to the symmetry group $GL(\infty)$ of the KP hierarchy. Moreover, we derive an explicit expression for this operator, up to a constant factor (which is checked to be equal to unity).
Since a relation of this type, which connects the Hodge tau-function with the Hurwitz tau-function, is known \cite{Kazarian},
formula (\ref{intM}) allows us describe all three tau-functions in terms of symmetry group operators.

The approach we use is based on the description of integrable hierarchies in terms of the Sato Grassmannian, and more specifically, in terms of the Kac--Schwarz operators. This construction allows us to derive linear constraints 
\begin{gather}\label{CONST}
\widehat{J}_m^{\alpha} \,\tau_\alpha =0, \,\,\,\, \mathrm{for}\,\,\,\, m\geq 1,\\
\widehat{L}_{m}^{\alpha}\, \tau_\alpha=0, \,\,\,\, \mathrm{for}\,\,\,\, m\geq -1, \label{CONST1}
\end{gather}
where the index $\alpha$ indicates one of three considered tau-functions. All operators $\widehat{J}_m^{\alpha}$ and $\widehat{L}_{m}^{\alpha}$ belong to the $\widehat{gl(\infty)}$ algebra (more specifically, to $W_{1+\infty}$), and
satisfy the commutation relations:
\be
\left[\widehat{J}_k^{\alpha},\widehat{J}_m^{\alpha}\right]_-=0,\,\,\,\, \mathrm{for}\,\,\,\, k,m\geq 1,\\
\left[\widehat{L}_{k}^{\alpha},\widehat{J}_m^{\alpha}\right]_-=-m\widehat{J}_{k+m}^{\alpha},\,\,\,\, \mathrm{for} \,\,\,\, k\geq -1 \,\,\,\, \mathrm{and} \,\,\,\, m\geq 1, \\
\left[\widehat{L}_{k}^{\alpha},\widehat{L}_{m}^{\alpha}\right]_-=(k-m)\widehat{L}_{k+m}^{\alpha}, \,\,\,\, \mathrm{for}\,\,\,\, k,m\geq -1.
\ee
The constraints for the Kontsevich--Witten tau-function are well known, namely, in this case equations (\ref{CONST}) describe a reduction from KP to KdV, and equations (\ref{CONST1}) are the Virasoro constraints. 
Two other families of constraints (for the Hurwitz and Hodge tau-functions) are obtained explicitly for the first time (see, however, \cite{MMKH,BM}). 

Although constraints for all three tau-functions satisfy the same commutation relations, they are quite different in their form. Namely, the constraints for the Kontsevich--Witten and Hodge tau-functions are given by the first and second order differential operators, while the constraints for the Hurwitz tau-function contain all higher derivative terms.

The present paper is organized as follows. Section \ref{S1} contains material on the classical KP hierarchy, tau-functions and Kac--Schwarz operators. In Section \ref{MAIN} we prove the relation between tau-functions and derive the constraints (\ref{CONST}) and (\ref{CONST1}). Section \ref{CONC} is devoted to concluding remarks.






\newpage
\section{KP hierarchy and its symmetries}\label{S1}
In this section we give a brief introduction to the subject, for more details, see \cite{JMbook,AZ,BBT,MorUFN,Kyoto,Fukuma} and references therein.

\subsection{Free fermions}

From the works of the Kyoto school \cite{Kyoto} it is known that infinite-dimensional groups act on the spaces of solutions of the integrable hierarchies. In particular, a central extension of the group $GL(\infty)$ acts on the space of solutions of the KP hierarchy. This action is very natural in the formalism of free fermions, which we are going to sketch in this section. In particular, the elements of the algebra $\widehat{gl(\infty)}$ (a central extension of a version of the $gl(\infty)$ algebra) is naturally given by bilinear normally ordered combinations of free fermions. 

Let us introduce the free fermions  $\psi_n , \psistar_{n}$, $n\in \mathbb{Z}$, which satisfy the canonical anticommutation relations
\be\label{anti}
[\psi_n , \psi_m ]_+ = [\psistar_n, \psistar_m]_+=0, \quad
[\psi_n , \psistar_m]_+=\delta_{mn}.
\ee
They generate
an infinite dimensional Clifford algebra. We use their
generating series
\be\label{ferm0}
\psi (z)=\sum_{k\in \z}\psi_k z^k, \quad \quad
\psistar (z)=\sum_{k\in \z}\psistar_k z^{-k}.
\ee

Next, we introduce a vacuum state $\left |0\rbr$, which is
a ``Dirac sea'' where all negative mode states are empty
and all positive ones are occupied:
$$
\psi_n \rvac =0, \quad n< 0; \quad \quad \quad
\psistar_n \rvac =0, \quad n\geq 0.
$$
(For brevity, we call indices $n\geq 0$ {\it positive}.)
Similarly, the dual vacuum state has the properties
$$
\lvac \psistar_n  =0, \quad n< 0; \quad \quad \quad
\lvac \psi_n  =0, \quad n\geq 0.
$$
With respect to the vacuum $\rvac$, the operators $\psi_n$ with
$n<0$ and $\psistar_n$ with $n\geq 0$ are annihilation operators
while the operators $\psistar_n$ with $n<0$ and
$\psi_n$ with $n\geq 0$ are creation operators. 

The normal ordering $\normord (\ldots )\normord $ with respect
to the Dirac vacuum $\rvac$ is defined as 
follows: all annihilation operators
are moved to the right and all creation operators are moved to
the left, taking into account that the factor $(-1)$ appears 
each time two neighboring 
fermionic operators exchange their positions. 
For example:
$\normord  \psistar_{1}\psi_{1}\normord =
-\psi_{1} \psistar_{1}$, 
$\normord  \psi_{-1}\psi_{0}\normord =
-\psi_{0} \psi_{-1}$, 
$\normord  \psi_{2}\psistar_{1}\psi_1 \psistar_{-2}\normord =
\psi_{2}\psi_1 \psistar_{-2}\psistar_{1}$,
etc.
Under the sign of normal ordering, all fermionic 
operators $\psi_j$ and $\psistar_j$ anticommute.
In other words, it is wrong to use the commutation relations
of the Clifford algebra under the sign of normal ordering,
i.e., for example, $\normord \psistar_{1}\psi_{1}\normord 
\neq \normord (1-\psi_{1} \psistar_{1})\normord$.

We also introduce ``shifted'' Dirac vacua $\rvacn$ and $\lvacn$
defined as   
\begin{align}\label{vacdefr}
\rvacn = \left \{
\begin{array}{l}
\psi_{n-1}\ldots \psi_1 \psi_0 \rvac , \,\,\,\,\, n> 0,
\\ \\
\psistar_n \ldots \psistar_{-2}\psistar_{-1}\rvac , \,\,\,\,\, n<0,
\end{array} \right.
\end{align}
\begin{align}\label{vacdefl}
\lvacn = \left \{
\begin{array}{l}
\lvac \psistar_{0}\psistar_{1}\ldots \psistar_{n-1} , \,\,\,\,\, n> 0,
\\ \\
\lvac \psi_{-1}\psi_{-2}\ldots \psi_{n} , \,\,\,\,\, n<0.
\end{array} \right.
\end{align}
For them we have 
\begin{equation}
\begin{split}
 \psi_m \rvacn &=0, \quad m < n; 
\qquad 
\psistar_m \rvacn =0, \quad m \ge n, \\
\lvacn  \psi_{m}&=0 , \quad m \ge n; 
\qquad 
\lvacn  \psistar_{m}=0 , \quad m < n.
\end{split}
\end{equation}
and
\begin{equation}
\begin{split}
\psi_n \rvacn &= \left|n+1 \rbr,
\qquad  \psistar_n \left|n+1 \rbr = \left|n \rbr,
\\
 \lbr n+1 \right|\psi_n &= \lbr n \right|,
\qquad  
\lbr n \right|\psistar_n = \lbr n+1 \right| .
\end{split}
\end{equation}

It is useful to introduce the bare vacuum 
$\left |\infty \right >$ and totally occupied space $\left |-\infty \right >$:
\be\left
|\infty \right >=\dots\psi_{2}\psi_{1}\psi_{0}\ \rvac, \\
\left|-\infty \right >=\dots\psistar_{-3}\psistar_{-2}\psistar_{-1}\rvac.
\ee
For these states we have:
\be\label{infvac}
\rvac=\psistar_{0}\psistar_{1}\psistar_{2}\dots \left |\infty \right >,\\
\rvac=\psi_{-1}\psi_{-2}\psi_{-3}\dots \left |-\infty \right >.
\ee
With respect to the totally occupied state, all $\psistar_j$'s are the annihilation
operators while all $\psi_j$'s are the creation operators. With respect to the bare vacuum it is vice versa:
\be
\psistar_k\left|-\infty \right >=\psi_k \left|\infty \right >= 0, \,\,\,\,\, k \in \mathbb{Z}.
\ee

Bilinear combinations $\sum_{mn} B_{mn}\psistar_m \psi_n$
of the fermions, with certain conditions
on the matrix $B = (B_{mn})$, generate an
infinite-dimensional Lie algebra.\footnote{To obtain a well-defined algebra and group one have to impose some restrictions on the matrix $B$ (for example, it can have only finite number of nonzero diagonals, or only a finite number of nonzero elements below (above) the principal diagonal \cite{Bombay}). We do not impose any restrictions of this type, thus, generally speaking,  ``algebra elements'' and  ``group elements'' we consider do not belong to a well-defined algebra or group. Products or commutators of such elements can be divergent and, in principle, should be regularized. The simplest example of such regularization corresponds to the commutation relation between the current components (\ref{Heis}) (see, e.g., \cite{AZ}).  No divergences appear in our calculations, thus we ignore this subtlety in what follows.}  Exponentiating these expressions one obtains
an infinite dimensional group (a version
of $GL(\infty )$).
Sometimes 
it is more convenient to consider the algebra elements 
with normally ordered 
bilinear fermionic combinations
\be\label{algel}
X_B=\sum_{mn} B_{mn}\normord\psistar_m \psi_n\normord.
\ee
Corresponding group elements
\begin{equation}\label{gl}
G=\exp \Bigl (\sum_{i, k \in {\z }}B_{ik}\normord\psistar_i \psi_k\normord\Bigr )
\end{equation}
act on the bare vacuum and totally occupied space in a particularly simple way. Namely, these states are the eigenstates for all group elements:
\be\label{grextr}
G\left|-\infty \right >= \exp\left(\sum_{i<0} B_{ii}\right)\left|-\infty \right >,\\
G\left|\infty \right >=  \exp\left(-\sum_{i\geq 0} B_{ii}\right) \left|\infty \right >.
\ee

Excited states (over the vacuum 
$\rvac$) are obtained by filling some empty states
(acting by the operators $\psistar_j$)
and creating some holes (that is acting by the $\psi _j$'s). 
A particle 
carries the charge $-1$ while a hole 
carries the charge $+1$, so any state with a definite number 
of particles and holes has the definite charge.
Let us introduce a convenient basis
of states with definite charge in the fermionic Fock space ${\cal H}_{F}$.
The basis states $\left |\lambda , n\rbr$ are
parametrized by $n$ (the charge $n$
with respect to the vacuum state $\rvac$) and Young diagrams $\lambda$ in the following way.
Given a Young diagram $\lambda =
(\lambda_1 , \ldots , \lambda_{\ell})$ with $\ell =\ell (\lambda )$
nonzero rows, let
$(\vec \alpha |\vec \beta )=(\alpha_1, \ldots , \alpha_{d(\lambda )}|
\beta_1 , \ldots , \beta_{d(\lambda )})$ be the Frobenius notation
for the diagram $\lambda$. 
Here $d(\lambda )$ is the number of
boxes on the main diagonal and $\alpha_i =\lambda_i -i$,
$\beta_i =\lambda'_i -i$, where $\lambda'$ is the transposed
(reflected about the main diagonal) diagram $\lambda$. Then
\beq\label{lambda1}
\begin{array}{l}
\left |\lambda , n\rbr \equiv
\psistar_{n-\beta_1 -1}\ldots \psistar_{n-\beta_{d(\lambda )}\! -1}\,
\psi_{n+\alpha_{d(\lambda )}}\ldots \psi_{n+\alpha_1}\rvacn ,
\\ \\
\lbr \lambda , n \right |\equiv
\lvacn \psistar_{n+\alpha_1}\ldots \psistar_{n+\alpha_{d(\lambda )}}\,
\psi_{n-\beta_{d(\lambda )}\! -1}\ldots \psi_{n-\beta_1 -1} .
\end{array}
\eeq
For the empty diagram
$\left < \emptyset , n\right |=\lvacn$, 
$\left | \emptyset , n\right >=\rvacn$.

An operator 
\be\label{energy}
E=\sum_{k \in \z}  k \normord\psi_k \psistar_k \normord
\ee 
is the energy operator. It defines a gradation on the Clifford algebra:
\be\label{fermen}
\left[E,\psi_k\right]_-=k\psi_k,\\
\left[E,\psistar_k\right]_-=-k\psistar_k,
\ee
so that the energy of the operators $\psi_k$ and $\psistar_k$ are $k$ and $-k$ respectively. We say that an element of the Clifford algebra has positive (negative) energy, if it can be represented as a sum of monomials $a_{\vec{k}\vec{m}}\psi_{k_1}\dots \psi_{k_n}\psistar_{m_1}\dots\psistar_{n_1}$ with $\sum_{i}{k_i}-\sum_{j}m_j>0$ (respectively $<0$). Operators (\ref{algel}) with a strictly upper triangular matrix $B$ have positive energy
and annihilate the left vacuum 
\be
\lvac X_B =0.
\ee
Operators with a strictly lower triangular matrix $B$ have negative energy and annihilate the right vacuum
\be
X_B \rvac=0.
\ee

From the commutation relations
\begin{equation}\label{ferm3}
\left[X_B,\psi_n\right]=\sum_i B_{in} \psi_i,\,\,\,\,\,\, \left[X_B,\psistar_n\right]=-\sum_i B_{ni} \psistar_i,
\end{equation}
we see that
\beq
e^{X_B}\psi_n e^{-X_B}=\sum_i \psi _i \, R_{in}, \quad \quad
e^{X_B}\psistar_n e^{-X_B}=\sum_i (R^{-1})_{ni}\psistar _i,
\eeq
where $R=\exp(B)$.

Expectation values of 
group elements are the $\tau$-functions
of integrable hierarchies of nonlinear differential equations. This means that they obey an
infinite set of the Hirota bilinear equations. The tau-function of the KP hierarchy labeled by a group element (\ref{gl}) is a ratio of two correlation functions:
\begin{equation}\label{tau}
\tau_G ({\bf t})=\frac{\lvac e^{J_+ ({\bf t})}G\rvac}{\lvac G\rvac} .
\end{equation}
It depends on
the variables ${\bf t}=\{t_1, t_2, \ldots \}$, usually called times, through the
 linear combination $J_+({\bf t})=\sum_{k>0}t_k J_k$ of the operators
\beq\label{Jk}
J_k =\sum_{j\in \z}\normord \psi_j \psistar_{j+k}\normord
=\mbox{res}_z \Bigl ( z^{-1} \normord \psi (z)
z^{k}\psistar (z)\normord \Bigr ).
\eeq
They are the Fourier modes of the ``current operator''
$J(z)\equiv z^{-1}\normord \psi (z)\psistar (z)\normord $ and span the Heisenberg algebra
\beq\label{Heis}
\left[J_k, J_l\right]= k \delta_{k+l,0}.
\eeq
Operators $J_k$ with positive and negative $k$ have negative and positive energy respectively, so that
\beq
J_k\rvac=\lvac J_{-k}=0\quad \quad \quad \mbox{ for $k\geq0$}.
\eeq

Normalization (\ref{tau}) guarantees that $\tau({\bf 0})=1$. The bilinear identity
\begin{equation}\label{Hir}
\oint_{\infty} e^{\xi({\bf t-t'},z)}
\tau ({\bf t}-[z^{-1}])\tau ({\bf t'}+[z^{-1}])dz =0
\end{equation}
encode all nonlinear equations of the integrable KP hierarchy.
Here we use the standard
short-hand notations
$$
{\bf t}\pm [z]\equiv \bigl \{ t_1\pm   
z, t_2\pm \frac{1}{2}z^2, 
t_3 \pm \frac{1}{3}z^3, \ldots \bigr \}
$$
and
\be
\xi({\bf t},z)=\sum_{k>0} t_k z^k.
\ee
The first nontrivial term in the expansion of the l.h.s. of (\ref{Hir}) gives the KP equation
\be
\tau\tau_{1111}-4\tau_1\tau_{111}+3\left(\tau_{11}\right)^2+
3\tau\tau_{22}-3\left(\tau_2\right)^2
-4\tau\tau_{13}+4\tau_1\tau_3=0,
\ee
where $\tau_{i_1i_2\ldots}\equiv\frac{\p}{\p{t_{i_1}}}\frac{\p}{\p{t_{i_2}}}\ldots\tau$.
The second derivative of this equation with respect to $t_1$ gives the 
KP equation in its standard form
$$
3u_{22}=\left(4u_3-12 u u_1- u_{111}\right)_1,
$$
where $u=\frac{\p^2}{\p{t_1}^2} \log (\tau)$.
A KP tau-function, independent\footnote{Actually, a prefactor which is exponential of a linear combination of even times is allowed, as it does not affect the Hirota equations.} on even times $t_{2k}, k=1,2,\ldots$, is a tau-function of the KdV hierarchy.

The coherent
states $e^{J_-({\bf t})}\rvacn$ (where $J_-({\bf t})=\sum_{k<0} t_k J_k$) and $\lvacn e^{J_+({\bf t})}$
can be expanded
as linear combinations of the basis states $\left |\lambda , n\rbr$. 
The coefficients are the famous Schur polynomials.
This expansion is important since it provides a link 
between hierarchies of integrable equations and the theory
of symmetric functions. Given a Young diagram $\lambda = (\vec \alpha |\vec \beta )$,
one can introduce the Schur polynomials (or Schur functions) 
via the
Jacobi-Trudi formula:
\beq\label{schur2}
s_{\lambda}({\bf t})=\det_{i,j=1, \ldots , \ell (\lambda )}
h_{\lambda_i -i +j}({\bf t}),
\eeq
where $h_k({\bf t})$  are
the Schur polynomials for one-row diagrams
\be\label{elSchur}
\exp \left(\xi({\bf t},z)\right)= \sum_{k=1}^\infty h_k({\bf t}) z^k.
\ee
It holds
\beq\label{lambda2}
\displaystyle{
e^{J_-({\bf t})}\rvacn = \sum_{\lambda} (-1)^{b(\lambda )}
s_{\lambda}({\bf t})\left |\lambda , n\rbr},
\eeq
\beq\label{lambda3}
\displaystyle{
\lvacn e^{J_+({\bf t})}=\sum_{\lambda} (-1)^{b(\lambda )}
s_{\lambda}({\bf t})\lbr \lambda , n \right |},
\eeq
where 
the sums run over all Young diagrams $\lambda$ including the
empty one, and we have introduced $b(\lambda )=\sum_{i=1}^{d(\lambda )}(\beta_i +1)$. For the empty diagram
$s_{\emptyset}({\bf t})=1$, $b(\emptyset )=0$.

The Baker--Akhiezer (BA)
function and its adjoint are given by the Sato formulas
\begin{equation}\label{BA1}
\psi ({\bf t},z)= e^{\xi ({\bf t},z)}
\frac{\tau ({\bf t}-[z^{-1}])}{\tau ({\bf t})}
=\frac{\left<1\right| e^{J_{+}({\bf t})}
\psi (z) G\rvac}{\lvac e^{J_{+}({\bf t})}G\rvac},
\end{equation}
\begin{equation}\label{BA1a}
\psi^* ({\bf t},z)= e^{-\xi ({\bf t},z)}
\frac{\tau ({\bf t}+[z^{-1}])}{\tau ({\bf t})}
=\frac{\left<-1\right|  e^{J_{+}({\bf t})}
\psistar (z) G\rvac}{z\lvac e^{J_{+}({\bf t})}G\rvac}.
\end{equation}
In terms of the BA function and its adjoint, the
bilinear identity (\ref{Hir}) acquires the form
\begin{equation}\label{bi1b}
\oint _{{\infty}} \psi({\bf t},z)\psi^{*}({\bf t'},z)\, dz =0.
\end{equation}

\subsection{$W_{1+\infty}$ algebra}

The boson-fermion correspondence allows us to represent a bilinear combination of the fermions in terms of the bosonic operators $J_k$:
\be\label{ferm}
\normord \psi (y ) \psistar (z)\normord=\frac{z}{y-z}\Bigl (\normordboson 
e^{\phi (y )-\phi (z)}\normordboson \, - 1\Bigr ),
\ee
where
\be\label{field}
\phi(z)=\sum_{k=1}^\infty k^{-1}\left({J_{-k}} z^k -{J_k}z^{-k} \right)+J_0\log z+P.
\ee
Operators $J_0$ and $P$ commute with all other components $J_k$ and satisfy the commutation relation
$$
\left[J_0,P\right]_-=1.
$$
Since the operator $P$ drops out of (\ref{ferm}), both $J_0$ and $P$ play no role for our consideration of the KP hierarchy and can be ignored in what follows (but they are important for MKP or Toda hierarchies). The normal ordering for bosonic operators $\normordboson\dots\normordboson$ puts all $J_k$ with positive $k$ to the right of all $J_k$ with negative $k$. 

An expansion of (\ref{ferm}) for $y=z+\epsilon$ with small $\epsilon$ yields the algebra $W_{1+\infty}$, which is an important subalgebra in $\widehat{gl(\infty)}$: 
\be
\normordboson e^{\phi (z+\epsilon )-\phi (z)}\normordboson =1+\sum_{m=1}^\infty \frac{\epsilon^{m}}{(m-1)!}\,
\widetilde{W}^{(m)}(z)
\ee
so that the operators $\widetilde{W}^{(m)}(z)=\sum_{n=-\infty}^\infty z^{-n-m} \widetilde{W}^{(m)}_n $ are just normal ordered elementary Schur functions (\ref{elSchur})
of the variables $p_k=\frac{1}{k!} \p^k_z \phi(z)$, where $\p_z\equiv \frac{\p}{\p z}$
\be
\widetilde{W}^{(m)}(z)=(m-1)! \normordboson h_m({\bf p})\normordboson.
\ee
The same expansion of (\ref{ferm}) allows us to identify
\beq
\widetilde{W}^{(m+1)}(z)=- z^{-1}\normord  \psistar (z) \p_z^{m} \psi (z )\normord,
\eeq
or, in terms of the bosonic current,
\be
\widetilde{W}^{(m+1)}(z)=\frac{1}{m+1}\normordboson 
\left(J(z)+\p_z \right)^m J(z)\normordboson. 
\ee
From (\ref{fermen}) follows that operator
\be
\widetilde{W}^{(m+1)}_n=-\mbox{res}_z \Bigl ( z^{-1} \normord \psistar (z)
z^{m+n}\p^{m}\psi (z)\normord \Bigr )
\ee
has energy $-n$:
\be
\left[E,\widetilde{W}^{(m)}_n\right]_-=-n\, \widetilde{W}^{(m)}_n.
\ee

Commutation relations of the operators $\widetilde{W}^{(m)}_n$ with the fermionic fields (\ref{ferm0}) follow form the more general relations (\ref{ferm3}):
\be\label{commu}
\left[\widetilde{W}^{(m+1)}_n,\psi (z )\right]_- =z^{m+n} \p_z^m \psi(z),\\ 
\left[\widetilde{W}^{(m+1)}_n,\psistar (z )\right]_- =- z (-\p_z )^m z^{m+n-1}\psistar(z).
\ee

Sometimes it is more convenient to use another basis for the same $W_{1+\infty}$ algebra \cite{Fukuma}. Namely, we put 
\be\label{Wcor}
W^{(1)}(z)=\widetilde{W}^{(1)}(z)=\p_z \phi(z),\\
W^{(2)}(z)=\widetilde{W}^{(2)}(z)-\frac{1}{2}\p_z \widetilde{W}^{(1)}(z)=\frac{1}{2}\normordboson\left(\p_z \phi(z)\right)^2\normordboson,\\
W^{(3)}(z)=\widetilde{W}^{(3)}(z)-\p_z \widetilde{W}^{(2)}(z)+\frac{1}{6}\p_z^2\widetilde{W}^{(1)}(z)=\frac{1}{3}\normordboson\left(\p_z \phi(z)\right)^3\normordboson,\\
\dots
\ee

Algebra $W_{1+\infty}$ is a central extension of the algebra $w_{1+\infty}$ of diffeomorphisms on the circle. There are many different ways to identify elements
\be\label{ordd}
a=\sum_{i\in\z,j\geq0} a_{i,j} z^i \p_z^j
\ee 
of $w_{1+\infty}$ with operators from
$W_{1+\infty}$.  We identify an operator (\ref{ordd}) with
\be\label{aoper}
W_a\equiv \mbox{res}_z \left(z^{-1} \normord\psistar(z) a\, \psi(z)  \normord\right)
\ee
so that 
\be
\widetilde{W}^{(m)}_n=W_{-z^{m+n-1} \p_z^{m-1}}.
\ee
From (\ref{ferm3}) it follows that
\be\label{wvsa}
\left[W_a,\psi(z)\right]_-=-a\,\psi(z),\\
\left[W_a,\psistar(z)\right]_-=a^*\psi(z),
\ee
where $a^*\in w_{1+\infty}$ is the adjoint operator for which an identity
$$ 
\mbox{res}_z \left(z^{-1} f(z)\, a\, g(z)  \right)
=\mbox{res}_z \left(z^{-1} g(z)\, a^*\, f(z) \right)
$$
holds for any commuting $f(z)$ and $g(z)$ (in particular,  $(z^k\p_z^m)^*=z(-\p_z)^m z^{k-1}$).

The algebra $W_{1+\infty}$ is a central extension of $w_{1+\infty}$, thus
\be
\left[W_{a_1},W_{a_2}\right]_-=W_{[a_1,a_2]_-} +C_{a_1,a_2},
\ee
where $C_{a_1,a_2}$ is a central element commuting with all other operators in $W_{1+\infty}$. This term can be effectively described by the commutation relations for the generating functions (for example, functions of this type were considered in \cite{OP}):
\be\label{Rdif}
R_n(q)\equiv W_{z^n q^{z\p_z}} =\mbox{res}_z \left(z^{-1} \normord\psistar(z) z^n q^{z \p_z} \psi(z)  \normord\right).
\ee
These generating functions can be represented in terms of the bosonic field (\ref{field}) as follows:
\be
R_n(q)=\frac{1}{1-q} \left(\mbox{res}_z\left(z^{n-1}\normordboson 
e^{\phi (qz )-\phi (z)}\normordboson\right) - \delta_{n,0}\right).
\ee
An operator $R_n$ is an infinite linear combination of operators $\widetilde{W}^{(m)}_n$:
\be
R_n(e^\epsilon)=-\left(\widetilde{W}^{(1)}_n+ \epsilon\,\widetilde{W}^{(2)}_n+\frac{\epsilon^2}{2}(\widetilde{W}^{(3)}_n+\widetilde{W}^{(2)}_n)+\cdots\right).
\ee
For $pq\neq1$ the commutation relation between operators $R_n$ can be found by a direct calculation:
\be\label{Phicom}
\left[R_n(q),R_m(p)\right]_-=\left(q^m-p^n\right)\left(R_{n+m}(qp)+ \delta_{m+n,0}\frac{1}{1-qp}\right),
\ee
which, for $p=q^{-1}$ reduces to
\be\label{comcoin}
\left[R_n(q),R_m(q^{-1})\right]_-=\left(q^{-n}-q^m\right)J_{n+m}+n q^{-n}\delta_{m+n,0}.
\ee
Operators $R_n$ are important for the description of the Virasoro constraints for the Hurwitz tau-function in Section \ref{KSHH}.

\subsection{Heisenberg, Virasoro, and ${W}^{(3)}$ algebras}

The operators $W^{(k)}(z)$ for $k=1,2,3$ are particularly important for our construction. Let us consider them in more detail. Operator $W^{(1)}(z)$ coincides with the current $J(z)$:
\be
{W}^{(1)}(z)=\p_z {\phi} (z)=J(z)= \sum_{m \in \z} \frac{J_m}{z^{m+1}}
\ee
and its components 
\be\label{OP1}
J_k=W_{-z^k}
\ee
span the Heisenberg algebra (\ref{Heis}).

The operator
\begin{equation}
W^{(2)}(z)=\frac{1}{2}\normordboson \left(\p_z {\phi} (z)\right)^2\normordboson = \sum_{m\in \z} \frac{{L}_m}{z^{m+2}}
\end{equation}
generates the Virasoro algebra $Vir \subset W_{1+\infty}$ with central charge $c=1$.
This algebra is spanned by the operators
\be\label{OP2}
L_n=W_{-z^n\left(z\p_z+\frac{n+1}{2}\right)}
\ee
satisfying the commutation relations
\be\label{Vircom}
[L_k,L_m]_-=(k-m){L}_{k+m}+\frac{1}{12}\delta_{k,-m}(k^3-k).
\ee
The operator $L_0$ coincides with the energy operator (\ref{energy}). 

We introduce two subalgebras $Vir_\pm$ of the Virasoro algebra, spanned by $L_{k}$ with strictly positive and strictly negative indices. We will say that an operator belongs to the group ${\bf Vir}_{+}$ (${\bf Vir}_{-}$), if it is of the form $\exp \left(V \right)$ with $V\in Vir_{+}$ ($V\in Vir_{-}$).

Operators $L_k$ together with $J_k$ span the so-called Heisenberg--Virasoro algebra ${\mathcal V}$. Commutation relations of ${\mathcal V}$ are given by (\ref{Heis}), (\ref{Vircom}) and
\be\label{Vircur}
[L_k,J_m]_-=-m J_{k+m}.
\ee
In what follows, we also consider the Heisenberg--Virasoro group with elements of the form
\be
C \exp\left(\sum_{k \in \z} \left(a_k L_k +b_k J_k\right)\right),
\ee
where $C$ is a constant. 

The subalgebra ${\mathcal V}_+ \subset{\mathcal V}$, spanned by $J_k$ with $k\geq 1$ and $L_m$ with $m\geq -1$, has no central terms in the commutation relations (\ref{Heis}), (\ref{Vircom}) and (\ref{Vircur}). This subalgebra describes the sets of constraints for the tau-functions considered in Section \ref{MAIN}.

The operator  
\be
W^{(3)}(z)=\frac{1}{3}\normordboson \left(\p_z\phi(z)\right)^3\normordboson=\sum_{n\in \z}\frac{{M}_n}{z^{n+3}},
\ee
where
\be\label{OP3}
M_n=W_{-z^n\left(z^2 \p_z^2 +(n+2)z\p_z+\frac{1}{6}(n+1)(n+2)\right)},
\ee
together with $W^{(2)}(z)$ generates the $W^{(3)}$-algebra introduced in \cite{W3}.

From the commutation relations (\ref{Heis}) we have
\be
\lvac e^{J_+({\bf t})} J_k =
\begin{cases}
\displaystyle{\frac{\p}{\p t_k} \lvac e^{J_+({\bf t})}}\,\,\,\,\,\,\,\,\,\,\, \mathrm{for} \quad k>0,\\[15pt]
\displaystyle{0}\,\,\,\,\,\,\,\,\,\,\,\,\,\,\,\,\,\,\,\,\,\,\,\,\,\,\,\,\,\,\,\,\,\,\,\,\,\,\,\,\,\,\mathrm{for} \quad k=0,\\[15pt]
\displaystyle{-kt_{-k} \lvac e^{J_+({\bf t})}}\,\,\,\,\,\mathrm{for} \quad k<0.
\end{cases}
\ee
This implies that for any operator $W$, which is a combination of the bosonic current modes (\ref{Jk}) (in particular, for any operator from $W_{1+\infty}$) there exists an operator $\widehat{W}$, which acts in the space of functions of the times $t_k$, such that:
\beq
\widehat{W} \, \lvac e^{J_+({\bf t})} = \lvac e^{J_+({\bf t})} W. 
\eeq
Namely, we put
\be
\widehat{J}_k =
\begin{cases}
\displaystyle{\frac{\p}{\p t_k} \,\,\,\,\,\,\,\,\,\,\,\, \mathrm{for} \quad k>0},\\[15pt]
\displaystyle{0}\,\,\,\,\,\,\,\,\,\,\,\,\,\,\,\,\,\,\, \mathrm{for} \quad k=0,\\[15pt]
\displaystyle{-kt_{-k} \,\,\,\,\,\mathrm{for} \quad k<0,}
\end{cases}
\ee
so that
\be\label{boscur}
\widehat{J}(z)=\sum_{k=1}^\infty \left(k t_k z^{k-1}+\frac{1}{z^{k+1}}\frac{\p}{\p t_k}\right).
\ee

This identification allows us to represent the transformation given by the group multiplication
\be\label{groupmpl}
\tau_H ({\bf t}) \mapsto \tau_{GH} ({\bf t}),
\ee
in terms of the operators, acting on the functions of the times $t_k$ (at least for $G$, which is a group element of the $W_{1+\infty}$ algebra). From the definition (\ref{tau}) it follows that  
\be
\tau_{e^{W}H}({\bf t})=C e^{\widehat{W}}\tau_H({\bf t}),
\ee
where $C=\frac{\lvac H\rvac}{\lvac e^{W}H\rvac}$ is a constant. This constant is equal to unity if the operator $W$ has positive energy.

For example, for $W=\sum_{k<0}  a_k J_k$ and an arbitrary group element $H$ we have
\be
\tau_{e^WH}({\bf t})= \exp\left(\sum_{k>0}k a_{-k} t_k\right)\tau_{H}({\bf t}). 
\ee
For $W=\sum_{k>0}  a_k J_k$ we have
\be
\tau_{e^WH}({\bf t})=\frac{\lvac H\rvac}{\lvac e^{W}H\rvac}  \times \tau_{H}({\bf t+a}). 
\ee

The Virasoro subalgebra of $W_{1+\infty}$ is generated by the operators
\begin{equation}
\label{virfull}
\widehat{L}_m=\frac{1}{2} \sum_{a+b=-m}a b t_a t_b+ \sum_{k=1}^\infty k t_k \frac{\p}{\p t_{k+m}}+\frac{1}{2} \sum_{a+b=m} \frac{\p^2}{\p t_a \p t_b}, 
\end{equation}
which are counterparts of (\ref{OP2}).
For the $W^{(3)}$ algebra we have
\begin{multline}
\widehat{M}_k=\frac{1}{3}\sum_{a+b+c=-k}a\, b\, c\, t_a\, t_b\, t_c+\sum_{c-a-b=k}a\, b\, t_a\, t_b\, \frac{\p}{\p t_{c}}\\
+\sum_{b+c-a=k}a\, t_{a}\frac{\p^2}{\p t_b\p t_c}+\frac{1}{3}\sum_{a+b+c=k}\frac{\p^3}{\p t_a \p t_b \p t_c}.
\end{multline}
In particular,
\be
\widehat{{M}}_0=\sum_{i,j\geq 1} \left(i\, j\, t_i\,t_j\frac{\p}{\p t_{i+j}} +(i+j)t_{i+j}\frac{\p^2}{\p t_i \p t_j}\right)
\label{CAJ}
\ee
is the cut-and-join operator of the Hurwitz tau-function \cite{HHK,HHK1}.


\subsection{Miwa parametrization and Grassmannian}\label{MiwaGrass}

The Miwa parametrization is very convenient for various problems, in particular for matrix models. There are different ways to introduce the Miwa parametrization. All of them are combinations of four basic possibilities, corresponding to the sign combinations in
\be
t_k=\pm\frac{1}{k}  \Tr Z^{\pm k}.
\ee
Here $Z$ is a diagonal matrix $\diag(z_1,z_2,\ldots,z_N)$ for some finite $N$ and the signs are independent.
Let us start with the Miwa parametrization
\be
t_k=-\frac{1}{k}  \Tr Z^{-k}.
\ee
A tau-function in this parametrization is
\be
\tau\left(-\left[Z^{-1}\right]\right)\equiv\tau\left(t_k=-\frac{1}{k}\tr Z^{-k}\right).
\ee

From the boson-fermion correspondence it follows that a tau-function in this parametrization is given by the correlation functions
\be
\tau\left(-\left[Z^{-1}\right]\right)= \frac{\lvac e^{-J_+ ([z_1^{-1}])-\ldots-J_+ ([z_N^{-1}])}G\rvac}{\lvac G\rvac} \\
= \frac{\lvac\psistar_{0}\ldots\psistar_{N-1}\psi(z_1)\ldots\psi(z_N)  G\rvac}{\lvac G\rvac \prod_{i<j}(z_i-z_j)}. 
\ee
The Wick theorem (see, e.g., \cite{AZ}) allows us to rewrite this fermionic correlation function as a ratio of determinants:
\be
\tau(-\left[Z^{-1}\right])=\frac{\det_{i,j=1}^N \varphi_i(z_j)}{\Delta(z)}, 
\ee
where
\be\label{varphit}
\varphi_i(z)=\frac{\lvac \psistar_{i-1} \psi(z)  G\rvac}{\lvac G\rvac},
\ee
and
\be
\Delta(z)=\prod_{i<j}(z_j-z_i)
\ee
is the Vandermonde determinant.

From the anticommutation relations (\ref{anti}) it follows that
\be
\left [ \psistar_{i-1},  \psi(z)\right ]_+= z^{i-1},  
\ee
so that the functions $\varphi_i(z)$, which are usually called {\it basis vectors}, have the following expansion
\be\label{basis}
\varphi_i(z)=z^{i-1} - \frac{\lvac  \psi(z)  \psistar_{i-1} G\rvac}{\lvac G\rvac} =z^{i-1}+\sum_{k=1}^\infty \varphi_{i,k} z^{-k},
\ee
where
\be\label{basis1}
\varphi_{i,k}=\frac{\lvac \psistar_{i-1} \psi_{-k}  G\rvac}{\lvac G\rvac}. 
\ee

The set $\{\varphi_i(z)\}$ defines a subspace $\mathcal{W}$ of an infinite dimensional Grassmannian. Corresponding theory was introduced by M. Sato in \cite{Sato} and further developed by G. Segal and G. Wilson in \cite{Segal}. Convergence of a tau-function as a function of times is not important for us (we consider a tau-function as a formal series in times), thus,
we will focus on Sato's version of the construction. Let us consider the space $H=H_+\oplus H_-$, where the subspaces $H_-$ and $H_+$ are generated by negative and nonnegative powers of $z$ respectively. Then the Sato Grassmannian $\rm{Gr}$ consists of all closed linear spaces $\mathcal{W}\in H $, which are compatible with $H_+$. Namely, an orthogonal projection $\pi_+ : \mathcal{W} \to H_+ $ should be a Fredholm operator, i.e. both the kernel ${\rm ker}\, \pi_+ \in \mathcal{W}$ and the 
cokernel ${\rm coker}\, \pi_+ \in H_+$ should be finite-dimensional vector spaces. 
The Grassmannian $\rm{Gr}$ consists of components $\rm{Gr}^{(k)}$, parametrized by an index of the operator $\pi_+$. We need only the component $\rm{Gr}^{(0)}$, which corresponds to the Dirac vacuum $\rvac$. Moreover, we will consider only the big cell $\rm{Gr}^{(0)}_+$ of $\rm{Gr}^{(0)}$, which is defined by the constraint ${\rm ker}\, \pi_+ = {\rm coker}\, \pi_+=0$.

For any $\mathcal{W}$ from $\rm{Gr}^{(0)}_+$ there exists a basis $\{\varphi_i(z)\}$ such that the matrix relating $\{\pi_+(\varphi_i(z))\}$ with $z^i$ has a well-defined determinant. Any such basis we call {\it admissible}. It can be always transformed to the basis of the form 
\be\label{GB}
\varphi_i(z)=z^{i-1} + \sum_{k>1-i}\varphi_{i,k} z^{-k}.
\ee
Obviously, for each $\mathcal{W}$ there exists a unique basis for which $\varphi_{i,k}=0$ for $k<1$ (that is the basis of the form (\ref{basis})). This basis is called {\it canonical}. 

Let us denote a point of $\rm{Gr}^{(0)}_+$, corresponding to a group element $G$, by $\mathcal{W}_G$. Then, the Baker--Akhiezer function (\ref{BA1}) belongs to the space $\mathcal{W}_G$ for all values of ${\bf t}$, for which the corresponding tau-function is not equal to zero \cite{Segal}. Then, from the expansion (\ref{lambda3}) it follows that
\be
\left<\lambda, 1\right|\psi(z) G \rvac \in \mathcal{W}_G 
\ee
for arbitrary $\lambda$. Thus
 \be\label{arbX}
 \lvac X \psi(z) G \rvac \in \mathcal{W}_G 
 \ee
for any $X$ from the Clifford algebra such that the correlation function does not vanish.

For any group element $G=\exp \Bigl (\sum_{i, k \in {\z }}B_{ik}\normord\psistar_i \psi_k\normord\Bigr )$ the matrix (\ref{basis1}) defines a group element, which is equivalent to $G$ in the following sense. Consider
\be\label{cangr}
\widetilde{G} =\exp\left(\sum_{i,k=1}^\infty \varphi_{i,k}  \psistar_{-k} \psi_{i-1} \right),
\ee
where $\varphi_{i,k}$ are the coefficients of the canonical basis (\ref{basis}).
It can be shown (see, e.g., (3.37) in\cite{AZ}) that
\be\label{ntrvac}
G\rvac = \lvac G\rvac \times \widetilde{G} \rvac.
\ee
Relation (\ref{grextr}) implies that $\widetilde{G}\left|\infty\right >=\left|\infty\right >$, thus
\beq
\widetilde{G} \rvac = \widetilde{\psistar_{0}}\,\widetilde{\psistar_{1}}\,\widetilde{\psistar_{2}} \dots \left |\infty\right >,
\eeq
where
\beq\label{fermtr}
\widetilde{\psistar_{k}}\equiv\widetilde{G} \psistar_{k} \widetilde{G}^{-1} =  \psistar_{k} +\sum_{i=1}^\infty \varphi_{k+1,i}\psistar_{-i}=\mbox{res}_z \left(z^{-1}\varphi_{k+1}(z) \psistar(z)\right).
\eeq
We see that the fermionic operators $\widetilde{\psistar_{k}}$, which describe the state $G\rvac$, are defined by the canonical basis vectors (\ref{varphit}).

Let us consider an operator $W_a\in W_{1+\infty}$, related to some differential operator $a\in w_{1+\infty}$ by (\ref{aoper}). Then, an action of the operators from the algebra $W_{1+\infty}$ as well as the corresponding group elements on the space of tau-functions can be translated to the action of algebra elements $w_{1+\infty}$ and corresponding group elements on the Grassmannian \cite{Fukuma}. Indeed, relation (\ref{ntrvac}) yields
\be
e^{W_a} G\rvac = \lvac G\rvac\times e^{W_a}\widetilde{G}\rvac, 
\ee
where
\be
e^{W_a}\widetilde{G} \rvac =\widetilde{\psistar_{0}}[a]\,\widetilde{\psistar_{1}}[a]\,\widetilde{\psistar_{2}}[a]  \dots e^{W_a}\left |\infty\right >,
\ee
and, for any operator $\Psi$ we define
\be
\Psi[a]\equiv e^{W_a}\Psi e^{-W_a}.
\ee
From (\ref{grextr}) it follows that $e^{W_a}\left |\infty\right >$ is proportional to $\left |\infty\right >$. Taking into account (\ref{fermtr}) and (\ref{wvsa}) we have
\be
\widetilde{\psistar_{k}}[a]\equiv e^{W_a}\,\widetilde{\psistar_{k}}\,e^{-W_a} = \mbox{res}_z \left(z^{-1}\varphi_{k+1}(z) e^{W_a}\psistar(z)e^{-W_a}\right)\\
=\mbox{res}_z \left(z^{-1}\varphi_{k+1}(z) e^{a^*}\psistar(z)\right)\\
=\mbox{res}_z \left(z^{-1}  \psistar(z) e^a \varphi_{k+1}(z)\right).
\ee
This observation justifies the identification (\ref{aoper}) between two types of operators.

We see that the action of the group element $e^{W_a}$ is equivalent to the action of the operator 
$e^a$ on the set of basis vectors. The problem is that, in general, for operators $W_a$, which have components with non-positive energy, the vectors $e^a \varphi_n$ are not of the form (\ref{GB}), but are Laurent series infinite in both directions. In spite of this difficulty, these vectors can sometimes constitute an admissible basis.

This can also be shown in a slightly different way. Namely, let us denote by $\varphi^n_i(z)$ the canonical basis vectors corresponding to the group element $e^{W_a}G$. Then, by definition,
\be
\varphi^{n}_i(z)\equiv\frac{\lvac \psistar_{i-1} \psi(z)  e^{W_a} G\rvac}{\lvac e^{W_{a}} G\rvac}
=\frac{\lvac  \psistar_{i-1} e^{W_{a}}  \psi(z)[-a]\,  G\rvac}{\lvac e^{W_{a}} G\rvac},
\ee
where, from (\ref{wvsa}) we have
$$
\psi(z)[-a]\equiv e^{-W_a}\psi(z)e^{W_a}= e^a \psi(z). 
$$
Thus,  
\be
\varphi^{n}_i(z)= e^a \lvac X \psi(z) G \rvac
\ee
for 
$$
X=\frac{  \psistar_{i-1}e^{W_{a}}}{\lvac e^{W_{a}} G\rvac}. 
$$
From (\ref{arbX}) it follows that $\mathcal{W}_{e^{W_a}G} \subset e^a \mathcal{W}_G$. On the other hand
\be
e^a \varphi_i(z)=\frac{\lvac \psistar_{i-1} e^{-W_a} \psi(z)  e^{W_a} G\rvac}{\lvac G\rvac} =\lvac \widetilde{X} \psi(z)  e^{W_a} G \rvac
\ee
for 
$$
\widetilde{X}=\frac{  \psistar_{i-1}e^{-W_{a}}}{\lvac  G\rvac},
$$
so that $e^a \mathcal{W}_G \subset   \mathcal{W}_{e^{W_a}G}$. This means that $e^a \mathcal{W}_G$ coincides with $\mathcal{W}_{e^{W_a}G}$.

We have seen that the vectors $e^a \varphi_i(z)$ belong to $\mathcal{W}_{e^{W_a}G}$. Do they constitute an admissible basis? The answer depends on the energy of the operator $W_a$. If it has positive energy, then 
\be\label{newvec}
e^a \varphi_i(z)=\frac{\lvac \psistar_{i-1} e^{-W_a} \psi(z)  e^{W_a} G\rvac}{\lvac G\rvac}=\frac{\lvac \psistar_{i-1}[a] \psi(z)  e^{W_a} G\rvac}{\lvac e^{W_a} G\rvac}=\sum_{k=1}^{i} \gamma_k^i \varphi^n_k(z),
\ee
where 
\be
\gamma_k^i=\mbox{res}_z \left(z^{-k} e^a z^{i-1}\right).
\ee 
We see that $\gamma_i^i=1$, so (\ref{newvec}) are basis vectors, but not necessary the canonical ones. 
For $W_a$ with energy equal to zero we have the same expression as in (\ref{newvec}), but only diagonal elements of the matrix $\gamma_k^i$ are not equal to zero. Thus, $e^a \varphi_i(z)$ in this case is proportional to the basis vector $\varphi_i^n(z)$ with the coefficient of proportionality 
\be
\gamma_i^i=\mbox{res}_z \left(z^{-i} e^a z^{i-1}\right).
\ee
If the operator $W_a$ has components of negative energy then, in general, the vectors $e^a \varphi_i(z)$ do not constitute an admissible basis.


For another sign convention in the Miwa parametrization: 
\be
t_k=\frac{1}{k} \Tr Z^{-k}
\ee
 a tau-function is again given by the ratio of determinants
\be
\tau\left(\left[Z^{-1}\right]\right)= \frac{\lvac e^{J_+ ([z_1^{-1}])+\cdots+J_+ ([z_N^{-1}])}G\rvac}{\lvac G\rvac} \\
=\frac{\lvac\psi_{-1}\ldots\psi_{-N}\psistar(z_N)\ldots\psistar(z_1)  G\rvac}{ \det (Z) \lvac G\rvac \Delta(z)}\\
= \frac{\det_{i,j=1}^N \varphi_i^*(z_j)}{\Delta(z_k)}, 
\ee
where
\be\label{varphitcon}
\varphi_i^*(z)=z^{-1}\times\frac{\lvac \psi_{-i} \psistar(z)  G\rvac}{\lvac G\rvac}. 
\ee
These basis vectors define the orthocomplement $\mathcal{W}^\perp$ of the subspace $\mathcal{W}$. Obviously, they can also be expressed in terms of the matrix $\varphi_{i,k}$
\be
\varphi_i^*(z)=z^{i-1}- \sum_{k=1} \varphi_{k,i} z^{-k} 
\ee
and the adjoint BA function (\ref{BA1a}) belongs to $\mathcal{W}^\perp$. Repeating the argument for $\varphi_k(z)$ one can show that the action of the operator $\exp(W_a)$ on the space of tau-functions is equivalent to the action of the operator $\exp(-z^{-1}a^*z)$ on the sets of the adjoint basis vectors $\varphi_{k}^*(z)$.

In the next chapter we will work with the Miwa parametrization, which uses an inverse matrix variable
\be\label{MainMiwa}
t_k= \frac{1}{k}\Tr Z^k. 
\ee
We denote by $\Phi_i(z)$ the basis vectors in this parametrization, namely
\be
\Phi_i(z)\equiv \varphi_i^*(z^{-1})= z^{1-i}+\dots
\ee
For this choice of the sign convention we have the following expression for a tau-function:
\be
\tau\left(\left[Z\right]\right)=\frac{\det_{i,j} \Phi_i(z_j)}{\Delta(z^{-1})}.
\ee
Corresponding identification between operators from $W_{1+\infty}$ and operators from $w_{1+\infty}$, acting on families of the basis vectors $\Phi_i(z)$, is
\be\label{basop}
W_{z^k \p_z^m} \mapsto - (z^2\p_z)^{m} z^{-k}.
\ee
This allows us to introduce $W_{1+\infty}$ operators, which are the counterparts of the $w_{1+\infty}$ operators, acting in this Miwa parametrization. For operators (\ref{basop}) we define
\be\label{Yquant}
{Y}_{- (z^2\p_z)^{m} z^{-k}}\equiv {W}_{z^k \p_z^m}.
\ee
Thus, action of the operator $a\in w_{1+\infty}$ on the set of basis vectors $\Phi_k(z)$  is equivalent to the action of the operator
\be\label{Ydif}
Y_a=\mbox{res}_z \left(z^{-1} \normord\psi\left(z^{-1}\right)\, z^{-1}\, a\, z\,\psistar(z^{-1})  \normord\right)
\ee
from the algebra $W_{1+\infty}$ on the group element. The corresponding bosonic operators are defined by
\be\label{obdifgen}
\widehat{Y}_{\left(z^2\p_z \right)^m z^{k}}=\mbox{res}_z\left(z^{-k}\normordboson \frac{(\widehat{J}(z)+\p_z)^{m}}{m+1} \widehat{J}(z)\normordboson \right).
\ee


In particular, in this parametrization the operators (\ref{OP1}), (\ref{OP2}) and (\ref{OP3}) correspond to the following operators from $w_{1+\infty}$\footnote{The operators used in \cite{Kazarian} are related to our operators by a conjugation $w= z\,w^{Kaz} z^{-1}$, just because M. Kazarian uses another normalization of the basis vectors, $\Phi_k= z\,\Phi_k^{Kaz}$.}
\be\label{obdif}
J_k=Y_{j_k} \leftrightarrow j_k=z^{-k},\,\,\,\,\,\,\,\, k\neq 1,\\ 
L_k=Y_{l_k} \leftrightarrow l_k=z^{-k}\left(z \p_z-\frac{k+1}{2}\right),\\
M_k=Y_{m_k} \leftrightarrow m_k=z^{-k}\left(z^2 \p_z^2-k z \p_z +\frac{(1+k)(2+k)}{6}\right).
\ee
A constant from $w_{1+\infty}$ corresponds to the operator $J_0$, which we can identify with zero as far as we consider the KP hierarchy.

\subsection{Kac--Schwarz operators}\label{KSsec}

Let $a\in w_{1+\infty}$ be an operator such that
\be\label{KS1}
a \, {\mathcal W} \subset {\mathcal W}
\ee
for some point of the Grassmannian $\rm{Gr}^{(0)}_+$ . Then, for the corresponding tau-function it holds that
\be\label{KStau}
\widehat{W}_{a}\tau= C\, \tau
\ee
for some constant $C$. Indeed,
\be\label{KS}
\widehat{W}_{a}\tau=\frac{\lvac e^{J_+(\bf t)} W_a G \rvac}{\lvac  G \rvac}
\ee
and, from (\ref{ntrvac}) we have
\begin{multline}
W_a G \rvac = \lvac G\rvac  \times W_a \widetilde{G}\rvac = \lvac G\rvac \,
\Big(\sum_{i=0}^\infty \widetilde{\psistar_{0}}\widetilde{\psistar_{1}}\dots\widetilde{\psistar_{i-1}} \left[W_a,\widetilde{\psistar_{i}}\right]_- \widetilde{\psistar_{i+1}}  \dots \left |\infty\right >\\
+ \widetilde{\psistar_{0}}\widetilde{\psistar_{1}}\dots W_a\left |\infty\right >\Big).
\end{multline}
Since
\be
\left[W_a,\widetilde{\psistar_{i}}\right]_- =\mbox{res}_z \left(z^{-1}  \psistar(z) a\, \varphi_{k+1}(z)\right) 
\ee
is a linear combination of $\widetilde{\psistar_{i}}$ for any $a$ satisfying (\ref{KS1}), and the totally occupied space is the eigenstate of any algebra element (\ref{algel}), we have
\be
W_a G \rvac = C\, G \rvac
\ee
for some constant $C$. Thus,
\be
\widehat{W}_{a}\tau= C \frac{\lvac e^{J_+({\bf t})}G\rvac}{\lvac  G \rvac} =C \tau.
\ee

Operators $a$ satisfying (\ref{KS1}), or similar relations for $\mathcal W^\perp$, we call the Kac--Schwarz operators \cite{Kac}. Obviously, the Kac--Schwarz operators form an algebra. However, general properties of such an algebra for arbitrary KP solutions are unknown (see, e.g., \cite{AMSvM} and \cite{PM} for recent discussion).

\subsection{Virasoro group action}\label{VIR}
In this subsection we describe how the subgroups $\bf{Vir}_{\pm}$ of the Virasoro group act on different spaces important for our construction. In particular, we consider an action
on the space of functions of times (the main example here, of course, is an action on the space of tau-functions), on the Heisenberg--Virasoro algebra and an action of the corresponding groups of diffeomorphisms on the space $H$. 

With any operator $\exp\left(\sum a_k L_k\right)$ from either ${\bf Vir}_+$ or ${\bf Vir}_-$, according to the rule (\ref{obdif}), we identify an operator $\exp\left(\sum a_k l_k\right)$. This operator is defined in terms of the formal series $g(z)=\sum a_k z^{1-k}$:
\be\label{pferd}
\sum a_k l_k= g(z)\p_z-\frac{1}{2}g'(z)-z^{-1}g(z).
\ee
This series allows us to define two formal Laurent series in $z$:
\be\label{fdif}
f(z)\equiv \exp\left(\sum a_k l_{+k}\right)\, z\, \exp\left(-\sum a_k l_{+k}\right)
\ee
and
\be\label{ftilde}
\widetilde{f}(z)\equiv\exp\left(\sum a_k l_{-k}\right)\, z\, \exp\left(-\sum a_k l_{-k}\right)=\frac{1}{f^{-1}(z^{-1})}.
\ee
For an operator from ${\bf Vir}_+$ the series $f$ and $\widetilde{f}$ are of the form 
$$
f(z)=z+b_{-1}+b_{-2}z^{-1}+b_{-3}z^{-2}+b_{-4}z^{-3}+\cdots
$$
and
\be
\widetilde{f}(z)=z+b_{{-1}}{z}^{2}+ \left( b_{{-2}}+{b_{{-1}}}^{2} \right) {z}^{3}+
 \left( 3\,b_{{-2}}b_{{-1}}+b_{{-3}}+{b_{{-1}}}^{3} \right) {z}^{4}\\
 +
 \left( 4\,b_{{-3}}b_{{-
 1}}+2\,{b_{{-2}}}^{2}+6\,b_{{-2}}{b_{{-1}}}^{2}+b_{
{-4}}+{b_{{-1}}}^{4} \right) {z}^{5}+\cdots
\ee
For an operator from ${\bf Vir}_-$ we have
$$
f(z)=z+b_{1}z^{2}+b_{2}z^{3}+b_{3}z^{4}+b_{4}z^{5}+\cdots
$$
and
$$
\widetilde{f}(z)=z+b_{{1}}+{\frac {b_{{2}}-{b_{{1}}}^{2}}{z}}+{\frac {2\,{b_{{1}}}^{3}
-3\,b_{{2}}b_{{1}}+b_{{3}}}{{z}^{2}}}+{\frac {10\,b_{{2}}{b_{{1}}}^{2}
+b_{{4}}-5\,{b_{{1}}}^{4}-4\,b_{{3}}b_{{1}}-2\,{b_{{2}}}^{2}}{{z}^{3}}
}+\cdots
$$
Here $b_k$'s are polynomials in the coefficients $a_k$.

Notations we use can be confusing, because it is not always clear if we consider a product of two operators (one of which can be an operator of order zero), or an operation of the operator on the function (so that the result is a function). To avoid possible confusion, when necessary we will denote the product of two operators as $a \cdot b$, while an action of the operator on the function as $a \left[b\right]$. For example:
$$
f(z)= \exp\left(\sum a_k l_k\right)\cdot z\cdot \exp\left(-\sum a_k l_k\right)=\exp\left( g(z)\p_z\right)\left[z\right].
$$
We use this notation for operators of all types. 

Series $f(z)$ and $\widetilde{f}(z)$ play an important role in our constructions. Thus, in what follows we need 
\begin{lemma}\label{LEM1} 
For any constant $\alpha$ and a series $f(z)$ defined by (\ref{fdif}) we have
\be\label{lem11}
\exp\left(g(z)\p_z+\alpha g'(z)\right)=\left(f'(z)\right)^\alpha \cdot\exp\left(g(z)\p_z\right)
\ee
and
\be\label{lem2}\exp\left(g(z)\p_z-\alpha\, \frac{g(z)}{z^k}\right)=\begin{cases}
\displaystyle{\exp\left(\frac{\alpha}{k-1}\left(\frac{1}{f(z)^{k-1}}-\frac{1}{z^{k-1}}\right)\right)\cdot\exp\left(g(z)\p_z\right)
 \,\,\,\,\, \mathrm{for} \quad k\neq 1}\\[20pt]
\displaystyle{\left(\frac{z}{f(z)}\right)^\alpha\cdot\exp\left(g(z)\p_z\right) \,\,\,\,\,\mathrm{for} \quad k=1.}
\end{cases}
\ee
\end{lemma}

\begin{proof}
Both relations follow from the Baker--Campbell--Hausdorff formula.
From this formula we know that for any $h(z)$ and $g(z)$ we have
\be\label{BCH1}
\exp\left(\alpha\, h(z)+g(z)\p_z\right)=\exp(\alpha\, n(z)) \,\exp\left(g(z)\p_z\right),
\ee
where $n(z)$ is a formal Laurent series in $z$:
\be
\label{tildeh}
n(z)=\frac{e^{g(z)\p_z}-1}{g(z)\p_z}\left[h(z)\right]=h(z)+ \frac{1}{2} g(z)\, \p_z h(z)+\dots
\ee
In particular, if $h(z)=g'(z)$, then
\be\label{pr1}
n(z)= \frac{e^{g(z)\p_z}-1}{g(z)\p_z}\left[g'(z)\right]= \left(e^{g(z)\p_z}-1\right)\left[\log(g(z))\right]=\log(g(f(z)))-\log(g(z)).
\ee
Moreover, we have
\be
g(f(z))=e^{g(z)\p_z} \left[g(z)\right]= \left(e^{g(z)\p_z} \cdot g(z)\,\p_z \right)\left[z\right] \\
= \left(g(z)\p_z\cdot  e^{g(z)\p_z} \right)\left[ z \right]=g(z) f'(z)
\ee
so that 
\be
n(z)= \log (f'(z))
\ee
and (\ref{lem11}) follows from (\ref{BCH1}).

For (\ref{lem2}) we have instead of (\ref{pr1})
\be
n(z)=\frac{e^{g(z)\p_z}-1}{g(z)\p_z}\left[-\frac{g(z)}{z^k}\right],
\ee
which, for $k\neq 1$, is equal to
\be
n(z)=\left(e^{g(z)\p_z}-1\right)\left[ \frac{1}{(k-1)z^{k-1}}\right]=\frac{1}{k-1}
\left(\frac{1}{f(z)^{k-1}}-\frac{1}{z^{k-1}}\right),
\ee
and for $k=1$
\be
n(z)=\left(e^{g(z)\p_z}-1\right)\left[-\log(z)\right]=\log\left(\frac{z}{f(z)}\right),
\ee
which establishes (\ref{lem2}). This completes the proof.
\end{proof}

From this lemma and (\ref{pferd}) we immediately arrive at the following expression for the operator (\ref{pferd}): 
\be\label{viract}
\exp\left(\sum a_k l_k\right)=\frac{z}{f(z)}\sqrt{f'(z)}\exp\left(g(z)\p_z\right).
\ee
In what follows we will also use the formula (\ref{lem2}) for $k=4$:
\be\label{k4case}
\exp\left(g(z)\p_z-\frac{g(z)}{z^4}\right)=\exp\left(\frac{1}{3f(z)^3}-\frac{1}{3z^3}
\right)\exp\left(g(z)\p_z\right).
\ee

Let us show, how the groups ${\bf Vir}_{\pm}$ act on the algebra $W_{1+\infty}$. For the current $J(z)$ from the commutation relation (\ref{Vircur}) we have:
\be
\left[L_k,J(z)\right]_-=z^k\left(z\p_z +(k+1)\right) \left[J(z)\right].
\ee
Thus
\be
e^{\sum a_k L_k} \cdot J(z) \cdot e^{-\sum a_k L_k}=\exp\left({\sum a_k z^k\left(z\p_z +(k+1)\right)}\right)  \left[J(z)\right].
\ee
From the definition of the series $\widetilde{f}(z)$ (\ref{ftilde})  it follows that (see, e.g., \cite{Kharchev}):
\be
\exp\left({\sum a_k z^k\left(z\p_z +(k+1)\right)}\right)  \left[J(z)\right]=\widetilde{f}'(z) \,J(\widetilde{f}(z)).
\ee
For the generating series of the Virasoro algebra $W^{(2)}(z)$ from (\ref{Vircom}) we have
\be
\left[L_k,W^{(2)}(z)\right]_-=z^k\left(z\p_z +2(k+1)\right) \left[W^{(2)}(z)\right]+\frac{1}{12}(k^3-k)z^{k-2},
\ee
so that
\be
e^{\sum a_k L_k} \cdot W^{(2)}(z) \cdot e^{-\sum a_k L_k}=\left(\widetilde{f}'(z)\right)^2 \,W^{(2)}(\widetilde{f}(z))+\frac{e^D-1}{D}\left[ \sum a_k (k^3-k)z^{k-2}\right],
\ee
where
\be
D=\sum a_k z^k\left(z\p_z +2(k+1)\right).
\ee
More generally, under conjugation the operator $\normordboson J(z)^k\normordboson$ behaves like a $k$-differential:
\be
e^{\sum a_k L_k}  \normordboson J(z)^k\normordboson  e^{-\sum a_k L_k}=\left(\widetilde{f}'(z)\right)^k \normordboson J(\widetilde{f}(z))^k \normordboson+\cdots.
\ee

Of course, the above formulas are central extensions of the conjugation relations for the algebra $w_{1+\infty}$. For example, let us consider the operators $l_m$ from (\ref{obdif}).
Since
\be
\exp\left(\sum a_k z^{1-k} \p_z\right)\cdot \p_z \cdot \exp\left(-\sum a_k z^{1-k} \p_z \right)=\frac{1}{f'(z)}\p_z
\ee
we have
\be\label{virsm}
\exp\left(\sum a_k l_k\right)\cdot l_{m}\cdot \exp\left(-\sum a_k l_k\right)= r(z)\p_z-\frac{1}{2}r'(z)-z^{-1}r(z),
\ee
where $r(z)=f(z)^{1-m}/f'(z)$.

Let us show how an element of the Virasoro group ${\bf Vir}_+$ acts on an arbitrary function of times (not necessary a tau-function) $Z({\bf t})$: 
\be
e^{\sum a_k \widehat{L}_k} \left[ Z({\bf t})\right] = e^{\sum a_k \widehat{L}_k}\cdot  Z({\bf t})\cdot e^{-\sum a_k \widehat{L}_k} \cdot e^{\sum a_k \widehat{L}_k}\left[1\right].
\ee
Let us assume that the function $Z({\bf t})$ is given by a correlation function
\be
Z({\bf t})=\left<e^{\sum_{k>0}k\, t_k S_k}\right>
\ee
in some model with some commuting operators $S_k$. Then, since the operators $\widehat{L}_k$ for positive $k$ annihilate constants, for an operator from ${\bf Vir}_+$ we have
\be\label{postgf}
e^{\sum_{k>0} a_k \widehat{L}_k} \left[Z({\bf t})\right]=\left<\exp\left(\sum_{k>0} k\, S_k\, e^{\sum_{k>0} a_k \widehat{L}_k}\cdot t_k\cdot e^{-\sum_{k>0} a_k \widehat{L}_k}\right)\left[1\right]\right>\\
=\left<\exp\left({\sum_{k>0}\left(k\, t_k\, \widetilde{S}_k+\widetilde{S}_{-k}\frac{\p}{\p t_k}\right)}\right)\left[1\right]\right>\\
=\left<\exp\left({\sum_{k>0}k \left(t_k\, \widetilde{S}_k+
\frac{1}{2}\widetilde{S}_k\widetilde{S}_{-k}\right)}\right)\right>,\\
\ee
where
\be
\widetilde{S}_k=\mbox{res}_z\, \left(z^{-k-1}\,Q \right)
\ee
for
\be
 Q=\sum_{k=1}^\infty  S_k\,f(z)^k.
\ee
Here $f(z)$ is the series (\ref{fdif}). The last line of (\ref{postgf}) can also be represented as 
\be\label{dtr}
e^{\sum_{k>0} a_k \widehat{L}_k} \left[Z({\bf t})\right]=
\left<\exp\left(\sum_{k>0}k\, t_k\,  \mbox{res}_z\,\left( \frac{Q}{z^{k+1}}\right)-\frac{1}{2}\mbox{res}_z\,\left({ Q}\, \p_z\, {Q}_{-}\right)\right)\right>.
\ee
If $S_k=\frac{1}{k}\Tr X^k$ for some $X$ (in particular, it can be a matrix in some matrix integral), then one can say even more. In this case (\ref{dtr}) reduces to
\be\label{dtr2}
e^{\sum_{k>0} a_k \widehat{L}_k} \left<\exp\left({\sum_{k>0} t_k\, \Tr X^k}\right)\right>=\\
 \left<\sqrt{\det\frac{{{\widetilde{f}(X)}\otimes {1}-{1}\otimes{\widetilde{f}(X)}}}{{{X}\otimes {1}-{1}\otimes{X}}}}\times\det\left(\frac{X}{\widetilde{f}(X)}\right)^N\exp\left(\sum_{k>0} t_k\, \Tr \widetilde{f}(X)^k\right)\right>,
\ee
where $\widetilde{f}$ is given by (\ref{ftilde}).
We see that if the correlation function in the l.h.s. of (\ref{dtr2}) is a tau-function, then this formula gives an infinite dimensional family of tau-functions given by correlation functions with double-trace interaction.

Action of the group ${\bf Vir}_-$ can be derived from the action of ${\bf Vir}_+$ and from the following observation: for operators (\ref{Wcor}) with $m=1,2,3$ we have
\be
\widehat{W}_k^{(m)}({\bf t}) \,\left[e^{\sum_{k>0}k\, t_k S_k}\right]=\widehat{W}_{-k}^{(m)}({\bf S})\, \left[e^{\sum_{k>0}k\, t_k S_k}\right],
\ee
where an operator $\widehat{W}_k^{(m)}({\bf S})$ acts in the space of function of $S_k$'s. In particular, for the Virasoro operators we have an identity
\be
\exp\left({\sum_{k<0} a_k \widehat{L}_k({\bf t})} \right)\left[e^{\sum_{k>0}k\, t_k S_k}\right]=
\exp\left({\sum_{k>0} a_{-k} \widehat{L}_k({\bf S})}\right) \left[e^{\sum_{k>0}k\, t_k S_k}\right].
\ee
Thus, to describe an action of ${\bf Vir}_-$ we can take the expression (\ref{dtr}) for an action of ${\bf Vir}_+$, interchange $t_k$ and $S_k$, and substitute $f(z)$ with $\widetilde{f}(z)$:
\be\label{dtr3}
e^{\sum_{k<0} a_k \widehat{L}_k} Z({\bf t})=\exp\left(-\frac{1}{2}\mbox{res}_z\,{ P}\, \p_z\, {P}_{-}\right)\times
\left<\exp\left( \sum_{k>0}k\, S_k\,  \mbox{res}_z\, \left(\frac{{P}}{z^{k+1}}\right)\right)\right>,
\ee
where
\be
P=\sum_{k=1}^\infty t_k\widetilde{f}(z)^k. 
\ee
For the operators $S_k= \frac{1}{k}\Tr X^k$ this expression reduces to
\be\label{dtr4}
e^{\sum_{k<0} a_k \widehat{L}_k} Z({\bf t})=\exp\left(-\frac{1}{2}\mbox{res}_z\,{ P}\, \p_z\, {P}_{-}\right)\times\left<\exp\left({\sum_{k>0} t_k\, \Tr \left[\widetilde{f}(X)^k\right]_+}\right)\right>,
\ee
where we use the notation $[\dots]_+$ for the part, which contains only strictly positive degrees of $X$.

Since operators from the algebra ${ Vir}_-$ are of the first order, operators from ${\bf Vir}_-$ define a linear change of variables when act on an arbitrary function:
\be
e^{\sum_{k<0} a_k \widehat{L}_k} Z({\bf t}) =e^{-\frac{1}{2}\sum_{i,j} A_{ij}t_i t_j} \,Z\left(\widetilde{\bf t}\right),
\ee
where
\be
\widetilde{t}_k= \mbox{res}_z\, \left(z^{-k-1}\,P \right),
\ee
and
\be
A_{ij}=\mbox{res}_z\, \left({\widetilde{f}(z)^{i}}\,\p_z\left[ {\widetilde{f}(z)^{j}}\right]_-\right).
\ee
For tau-functions this transformation is known as a transformation between equivalent hierarchies \cite{Shiota,EqH,Kharchev}. 

\subsection{Matrix models}
Any formal series in an infinite set of variables $t_k$ can be expanded in a sum of the Schur functions. Expansions of the tau-functions of the KP hierarchy are quite special, namely the coefficients $c_\lambda$ parametrized by the Young diagrams $\lambda$:
\be
\tau({\bf t})=\sum_\lambda c_\lambda s_\lambda({\bf t})
\ee
satisfy the Pl\"{u}cker relations \cite{JMbook}. From the expansion (\ref{lambda3}) the fermionic correlation function expression for the coefficients easily follows:
\be
c_{\lambda}=(-1)^{b(\lambda )}\left<\lambda,0\right |G\rvac.
\ee
Let us denote by $\tau^N(t)$ a restricted sum 
\be\label{rest1}
\tau^N({\bf t})=\sum_{l(\lambda)\leq N} c_\lambda s_\lambda({\bf t}),
\ee
where $l(\lambda)$ is the length of the partition,
so that $\tau({\bf t})=\tau^\infty({\bf t})$.
This restricted sum is also a KP tau-function for any $N$ (see, e.g., \cite{AZ}). 
Since for any $N\times N$ matrix $Z$ the Schur function labeled by $\lambda$  vanishes if $l(\lambda)>N$
\be
s_\lambda \left(t_k=\frac{1}{k} \tr Z^k\right)=0,
\ee
one has 
\be\label{rest2}
\tau^N\left(\left[Z\right]\right)=\tau^{N+1}\left(\left[Z\right]\right)=\cdots =\tau\left(\left[Z\right]\right).
\ee
The restricted sums (\ref{rest1}) naturally appear in the expansion of the matrix integrals. Let us consider classes of the matrix models, which are most important for our purposes. 

Unitary matrix integrals are of primary interst for us.  
We use the Haar measure normalized in such a way that the integral over the unitary group is equal to unity:
\be
\int_\mathcal{U} \left[d {U}\right]=1.
\ee
Integration rules for the Schur functions are particularly simple:
\be
\int_\mathcal{U} \left[d{U}\right] s_\lambda\left({ {\left[UAU^\dagger B\right]}}\right)=\frac{s_\lambda({\left[ A\right]})\chi_\lambda({\left[ B\right]})}{\dim_\lambda},
\label{ir1}
\ee
\be
\int_\mathcal{U} \left[d{ U}\right] s_\lambda\left(\left[{UA}\right]\right)s_\mu\left(\left[U^\dagger B\right]\right)=\frac{s_\lambda\left(\left[{{A B}}\right]\right)}{\dim_\lambda}\delta_{\lambda,\mu},
\label{ir2}
\ee
where $\dim_\lambda$ is a value of the Schur polynomial in the Miwa parametrization with the unity matrix $I=\diag(1,1,\dots,1)$:
\be
\dim_\lambda=s_\lambda\left(\left[I\right]\right)=\prod_{0<i<j\leq N}\frac{\lambda_i-\lambda_j+j-i}{j-i}.
\ee
With the help of the Cauchy--Littlewood identity, this leads us to the following expansion of the unitary matrix model:
\be
\int_\mathcal{U}\left[d{U}\right]\exp\left(\sum_{k=0}^\infty t_k\Tr { U}^k+\bar{t}_k \Tr { U^\dagger}^k\right)=\sum_{l(\lambda)\leq N} s_{\lambda}(t)s_\lambda(\bar{t}).
\label{unmm2}
\ee

To restore the explicit time dependence for any tau-function in the Miwa parametrization one can use the unitary matrix integral:  
\be\label{Uni}
\tau^N({ \bf t})=\int_\mathcal{U} \left[d U\right] \exp\left(\sum_{k=1}^\infty t_k \Tr U^{\dagger k}\right ) \tau\left(\left[U\right]\right).
\ee

The Itzykson--Zuber (IZ) integral for diagonal matrices $A$ and $B$ is a simple symmetric combination of the eigenvalues of $A$ and $B$ :
\be
\int_\mathcal{U} \left[d{U}\right] \exp\left(\Tr\left({{{U}A{U^\dagger }B}}\right)\right)=\left(\prod_{k=1}^{N-1}k!\right)\frac{\det e^{a_ib_j}}{\Delta(a)\Delta(b)}.
\label{IZ2}
\ee
In what follows we will mostly work with the eigenvalue integrals. For example, for ${A}={B}={I}$ the orthogonality condition (\ref{ir2}) in terms of eigenvalues reduces to:
\be
\prod_{j=1}^N\frac{1}{2\pi i}\oint_{|u_j|=1} \frac{d u_j}{u_j} \left|\Delta(u)\right|^2 s_{\lambda}(u)s_{\mu}(\bar{u})=N!\delta_{\lambda,\mu}.
\label{eigort}
\ee

Another important class of matrix integrals is given by integrals over Hermitian matrices. A Hermitian matrix can be decomposed into the product ${\Phi}={U}X{ U^\dagger}$ with unitary ${U}$ and real diagonal $X$. Then the element of the volume in the Hermitian matrix integral $\int_{\mathcal H}[d{ \Phi}]\dots$ is:
\be
[d{ \Phi}]=\left[d {U}\right] \Delta(x)^2 \prod_{i=1}^N d x_i.
\label{hermes}
\ee

The normal matrix integral is an integral over normal matrices (that is over matrices commutating with their conjugate, $\left[{Z},{ Z^\dagger}\right]=1$). 
A normal matrix can be diagonalized
\be
{Z}={U}Z{U^\dagger},
\ee
with the unitary matrix $ U$ and the diagonal matrix $Z$ with complex entries. Then the measure in the normal matrix integral $\int_{\mathcal N}\left[d { Z}\right]\dots$ is
\be\label{NMMm}
\left[d { Z}\right]= \left[d{ U}\right] \left|\Delta(z)\right|^2\prod_{i=1}^N d^2 z_i.
\ee




\newpage
\section{Three tau-functions and relations between them}\label{MAIN}

In this section we investigate the properties of several generating functions of enumerative geometry and relations between them.  All considered partition functions are tau-functions of the KP integrable hierarchy, and, in what follows, we focus on their integrable properties.

\subsection{Tau-functions and enumerative geometry}

 This section is devoted to a brief reminder of the geometric origin of the considered tau-functions, for more details see, e.g.,\cite{Kazarian,ELSV,Mul,GJV}.

The first of the partition functions we consider is the generating function of linear Hodge integrals. Let $\overline{\cal M}_{p;n}$ be the Deligne--Mumford compactification of the moduli space of stable complex curves with $n$ marked points. We consider linear Hodge integrals 
\be
\int_{\overline{\cal M}_{p;n}}\lambda_j \psi_1^{m_1}\psi_2^{m_2}\cdots\psi_n^{m_n}=\left<\lambda_j \tau_{m_1}\ldots\tau_{m_n} \right>,
\ee
where  $\psi_i$ is the first Chern class of the line bundle corresponding to the cotangent space of the curve at the $i$-th marked point and $\lambda_i$ is the $i$-th Chern class of the Hodge bundle. These integrals are trivial, unless the corresponding complex dimensions coincide:
\be
j+\sum_{i-1}^n m_i =\dim\left(\overline{\cal M}_{p;n}\right),
\ee
where $\dim\left(\overline{\cal M}_{p;n}\right)=3p-3+n$.

Let us introduce the generating function of linear Hodge integrals:
\be\label{tildegf}
\tilde{F}({\bf t^o};u)=\sum_{j\geq 0}(-1)^j \left<\lambda_j \exp\left(\sum_{k \geq 0} (2k+1)!!\, t_{2k+1} \tau_k \right)\right>u^{2j},
\ee
where $\lambda_0=1$ and ${\bf t^o}$ denotes a set of odd times $t_{2k+1}$. \footnote{Here we want to stress that the chosen normalization of the variables $t_k$ does not coincide with the one generally accepted in the enumerative geometry, but is natural for matrix models and integrable systems.} The change of variables
\be
T_1({\bf t})=t_1,\\
T_{2k+3}({\bf t})=\frac{1}{2k+3}\sum_{m\geq 1} \left(m\,u^2 t_m+2(m+1)u\, t_{m+1}+(m+2)t_{m+2}\right) \frac{\p}{\p t_m}
 T_{2k+1}({\bf t})\\
 =\frac{1}{2k+3}\left(u^2 \widehat{L}_0+2u\,\widehat{L}_{-1}+\widehat{L}_{-2}\right)T_{2k+1}({\bf t}),
\ee
such that $T_{2k+1}({\bf t})=t_{2k+1}+O(u)$, allowed M. Kazarian to relate the generating function (\ref{tildegf}) (which is a solution of integrable hierarchy of topological type \cite{Buryak} but not a solution of the KP) to the KP hierarchy. Namely, in \cite{Kazarian} it was proved that the exponential of the function
\be
F_{Hodge}({\bf t};u)\equiv \tilde{F}({\bf T^o}({\bf t});u)
\ee
is a tau-fucntion of the KP hierarchy for arbitrary $u$:
\be\label{gf}
\tau_{Hodge}({\bf t};u)\equiv\exp\left(F_{Hodge}({\bf t};u)\right).
\ee
This function, as opposed to (\ref{tildegf}), depends on both even and odd times. 

For $u=0$ only $\psi$-classes survives in (\ref{tildegf}):
\be
F_{KW}({\bf t^o})\equiv \left<\exp\left(\sum_{k>0} (2k+1)!!\, t_{2k+1} \tau_k \right)\right>=F_{Hodge}({\bf t};0)=\tilde{F}({\bf t^o};0)
\ee
and (\ref{gf}) reduces to the profound Kontsevich--Witten (KW) tau-function 
\be\label{KW}
\tau_{KW}({\bf t^o})=\exp\left(F_{KW}({\bf t^o})\right)=\tau_{Hodge}({\bf t};0).
\ee

It is known that linear Hodge integrals can be expressed through the intersection numbers of the $\psi$-classes \cite{Mumford,FP}. Namely,
\be
e^{\tilde{F}({\bf t^o};u)}
=e^{\widehat{Q}}  \tau_{KW}({\bf t^o}), 
\ee
where
\be
\widehat{Q}=\sum_{k=1}^\infty \frac{B_{2k}u^{4k-2}}{2k(2k-1)}\widehat{Q}_k.
\ee
Here
\be
\widehat{Q}_k=\sum_{i\geq 0} \frac{(2i+1)!!}{(2i+4k-1)!!}\tilde{t}_{2i+1}\frac{\p}{\p t_{2i+4k-1}}-\frac{1}{2}\sum_{i+j=2k-2}\frac{(-1)^i}{(2i+1)!!(2j+1)!!}\frac{\p^2}{\p t_{2i+1}\p t_{2j+1}}
\ee
and $\tilde{t}_k=t_k-\frac{\delta_{k,3}}{3}$ are the times subject to the dilation shift and $B_{2k}$ are the Bernoulli numbers 
$$
\frac{xe^x}{e^x-1}=1+\frac{x}{2}+\sum_{k=1}^{\infty}\frac{B_{2k}x^{2k}}{(2k)!}.
$$
Operator $\widehat{Q}$ does not belong to the $\widehat{gl(\infty)}$ symmetry algebra of the KP hierarchy. 

Hurwitz numbers count ramified coverings of Riemann surfaces. More specifically, the Hurwitz number $h(p|m_1,\ldots,m_n)$ gives the number of the Riemann sphere coverings with $N$ sheets, $M$ fixed simple ramification points and a single point with ramification structure given by $\{m_i\}$, a partition of $N$. The number of double ramification points $M$, the genus $p$ of the cover and the partition $\{m_i\}$ are related:
\be
M=2p-2+\sum_{i=1}(m_i+1).
\ee\label{dimcon2}
The generating function of the Hurwitz numbers
\be
F_H({\bf t};\beta)=\sum_{n>1}\frac{1}{n!}\sum_{p;m_1,\ldots,m_n}\frac{h(p;m_1,\ldots,m_n)}{M!} \beta^M m_1\ldots m_n
 t_{m_1}\ldots t_{m_n} 
\ee
defines the Hurwitz tau-function
\be\label{Hurwitzgf}
\tau_H({\bf t};\beta)=\exp\left(F_H({\bf t};\beta)\right).
\ee
$\tau_H$ is known to be a tau-function of the KP hierarchy (moreover, its generalization for double Hurwitz numbers is a tau-function of the 2D Toda lattice \cite{TodaOk}).

The Ekedahl, Lando, Shapiro, and Vainshtein (ELSV) formula \cite{ELSV} relates the Hurwitz numbers $h(p;m_1,\ldots,m_n)$ with linear Hodge integrals
\be\label{ELSV}
\frac{h(p;m_1,\ldots,m_n)}{M!}=\prod_{i=1}^n\frac{m_i^{m_i}}{m_i!}\int_{\overline{\cal M}_{p;n}}\frac{1-\lambda_1+\lambda_2-\ldots\pm\lambda_p}{\prod_{i=1}^n(1-m_i \psi_i)}.
\ee
This formula allowed M. Kazarian \cite{Kazarian} to find a relation between the Hurwitz tau-function (\ref{Hurwitzgf}) and the Hodge tau-function (\ref{gf}). These two tau-functions are related with each other by the $\widehat{GL(\infty)}$ group element. Our goal is to extend this connection and to include the KW tau-function into it.

\subsection{Heisenberg--Virasoro group and three tau-functions}\label{Con}

From \cite{Kazarian,MMKH} we know that the ELSV formula allows us to connect the generating function for the Hodge integrals (\ref{gf}) and the Hurwits tau-function (\ref{Hurwitzgf}) in a simple way: 
\be\label{Kaz}
\tau_{Hodge}({\bf t},u)=\widehat{G}_0\, \widehat{G}_- \,\tau_H({\bf t},\beta),
\ee
where
\be\label{HtoH}
\widehat{G}_-=e^{\widehat{L}_-} e^{-\sum_{k>0}\frac{k^{k-1}\beta^{k-1}t_k}{k!}},\\
\widehat{G}_0=\beta^{-\frac{4}{3}\widehat{L}_0} 
\ee
are the elements of the Heisenberg--Virasoro group. 
In what follows we put $\beta=u^3$. The operator $\widehat{L}_-$ belongs to the Virasoro algebra $Vir_-$ and is described below. 


In \cite{Fromto} we have claimed that the relation (\ref{Kaz}) can be naturally extended to include the KW tau-function. Here we clarify this extension. Namely 
\begin{conj}\label{CONJ}
\be\label{FromKWtoH}
\tau_{KW}({\bf t^o})=\widehat{G}_+\,\tau_{Hodge}({\bf t},u),
\ee
where
\be\label{Gplus}
\widehat{G}_+=\beta^{-\frac{4}{3}\widehat{L}_0}e^{\widehat{L}_+}\beta^{\frac{4}{3}\widehat{L}_0} \,\,\, \in {{\bf Vir}_+},
\ee
and operator $\widehat{L}_+$ is defined by the series $f_{+}$ from (\ref{fminus}). 
\end{conj}
In Section \ref{PR1} we prove (\ref{FromKWtoH}) up to a constant prefactor. Namely,  
\begin{theorem}\label{THEOR}
\be
\tau_{KW}({\bf t^o})=C(u)\, \widehat{G}_+\,\tau_{Hodge}({\bf t},u),
\ee
where $C(u)$ is a Taylor series in $u$ with constant coefficients of the form
\be
C(u)=1+\sum_{k=1}^\infty c_k u^{6k}.
\ee
\end{theorem}
Explicit calculations show that, at least, $C(u)=1+O(u^{30})$.

From Conjecture \ref{CONJ} it follows that the KW and Hurwitz tau-functions are related by an operator from the Heisenberg--Virasoro group:
\be\label{master}
\tau_{KW}({\bf t^o})=\widehat{G}_+\,\widehat{G}_0\,\widehat{G}_- \tau_H({\bf t},\beta).
\ee
Let us describe the operators $\widehat{L}_\pm$ in more detail. Namely,
\be
\widehat{L}_\pm=\sum_{k>0}a_{\pm k} \beta^{\mp k} \widehat{L}_{\pm k}
\ee
belong to the algebras $Vir_{\pm}$.
Coefficients $a_{\pm k}$ can be described by two formal Laurent series (see Section \ref{VIR})
\be
f_\pm(z)\equiv \exp\left(\sum_{k>0} a_{\pm k} z^{1\mp k} \p_z \right) z.
\ee
These two series are of different complexity: while $f_-(z)$ is relatively simple
\be\label{fplus}
f_-(z)=\frac{z}{1+z}e^{-\frac{z}{1+z}}=z-2z^2+O(z^3),
\ee
with an inverse series
\be
f_-^{-1}(z)=\sum_{k=1}^\infty\frac{k^k}{k!}z^k,
\ee
the series $f_+(z)$ is given implicitly as a solution of the equation
\be\label{fminus}
\frac{f_+(z)}{1+f_+(z)}\exp\left({-\frac{f_+(z)}{1+f_+(z)}}\right)=E \exp(-E),
\ee
where
\be
E=1+\sqrt{\left(\frac{1}{1+f_+(z)}\right)^2+\frac{4}{3z^3}}.
\ee
The solution of the equation (\ref{fminus}) is uniquely specified by the asymptotics for large $|z|$:
\be
f_+(z)=z-\frac{2}{3}+O\left(z^{-1}\right).
\ee

Let us stress that both $f_+$ and $f_-$ can be represented as compositions of two intermediate series:
\be
f_\pm(z)=f_{\pm 1}\left(f_{\pm 2}(z)\right),
\ee
where
\be
f_{\pm i}(z)\equiv \exp\left(\sum_{k>0} a^{(i)}_{\pm k} z^{1\mp k} \p_z \right) z 
\ee
for $i=1,2$.
This factorization corresponds to the factorization of the Virasoro group operators:
\be
\exp\left(\widehat{L}_{\pm}\right)=\exp\left(\widehat{L}_{\pm}^{(2)}\right)\cdot\exp\left(\widehat{L}_{\pm}^{(1)}\right).
\ee
For (\ref{fplus}) the factorization is obvious:
\be\label{minfunc}
f_{-1}(z)=z\, e^{-z}\\
f_{-2}(z)=\frac{z}{1+z}.
\ee
Factorization of the series $f_+$ is less trivial. Namely, $f_+$ can be expressed as a composition of
\be\label{negop}
f_{+1}(z)=\frac{1}{z\exp\left(z^{-1}\right)\sinh\left(z^{-1}\right)-1}\\
=z-\frac{2}{3}+\frac{1}{9}\,{z}^{-1}+{\frac {2}{135}}\,{z}^{-2}-{\frac {1}{405}}\,{z}^{
-3}-{\frac {2}{1701}}\,{z}^{-4}-{\frac {2}{127575}}\,{z}^{-5}+{\frac {
4}{54675}}\,{z}^{-6}\\
+{\frac {13}{1148175}}\,{z}^{-7}-{\frac {614}{
189448875}}\,{z}^{-8}-{\frac {958}{795685275}}\,{z}^{-9}+{\frac {668}{
14105329875}}\,{z}^{-10}+O \left( {z}^{-11} \right) 
\ee
and $f_{+2}$, satisfying the equation
\be
\frac{1}{\left(f_{+2}(z)\right)^{2}}\coth\left(\frac{1}{f_{+2}(z)}\right)-\frac{1}{f_{+2}(z)}=\frac{1}{3z^3}.
\label{hmin}
\ee
The series
\be
f_{+2}(z)=z-\frac{1}{45}{z}^{-1}+{\frac {1}{1575}}\,{z}^{-3}+{\frac {1}{273375}}\,{z}^
{-5}-{\frac {1658}{1326142125}}\,{z}^{-7}+{\frac {251}{15962821875}}\,
{z}^{-9}\\
+{\frac {1952908}{523660371609375}}\,{z}^{-11}-{\frac {
10945903}{80120036856234375}}\,{z}^{-13}+O \left( {z}^{-15} \right)
\ee
contains only odd terms, so that the Virasoro constraints for the KW tau-function allow us to get rid of the corresponding operator (see Section \ref{Sim}).

\subsection{Kontsevich--Witten tau-function}
The Kontsevich--Witten tau-function \cite{Konts,Witten}  is one of the most important objects of modern mathematical physics. It is given by a formal series in times with rational coefficients:
\be
\tau_{KW}\left({\bf t}\right)=1+\frac{1}{6}\,{t_{{1}}}^{3}+\frac{1}{8}\,t_{{3}}
+{\frac {1}{72}}\,{t_{{1}}}^{6}+{\frac {25}{48}}\,t_{{3}}{t_{{1}}}^{3}+
{\frac {25}{128}}\,{t_{{3}}}^{2}+\frac{5}{8}\,t_{{5}}t_{{1}}
+{\frac {1}{1296}}\,{t_{{1}}}^{9}\\+{\frac {49}{576}}\,{t_{{1}}}^{6}t_{{3
}}+{\frac {1225}{768}}\,{t_{{1}}}^{3}{t_{{3}}}^{2}+{\frac {35}{48}}\,{
t_{{1}}}^{4}t_{{5}}+{\frac {1225}{3072}}\,{t_{{3}}}^{3}+{\frac {245}{
64}}\,t_{{5}}t_{{3}}t_{{1}}+{\frac {35}{16}}\,{t_{{1}}}^{2}t_{{7}}+{
\frac {105}{128}}\,t_{{9}}+\dots
\ee
In the Miwa parametrization it is equal to the asymptotic expansion of the Kontsevich matrix integral over the Hermitian matrix $\Phi$:
\be
\tau_{KW}\left(\left[\Lambda^{-1}\right]\right)=\frac{\displaystyle{\int_{\mathcal H} \left[d \Phi\right]\exp\left(-{\Tr\left(\frac{\Phi^3}{3!}+\frac{\Lambda \Phi^2}{2}\right)}\right)}}{\displaystyle{\int_{\mathcal H} \left[d \Phi\right]\exp\left(-{\Tr\frac{\Lambda \Phi^2}{2}}\right)}}.
\label{matint}
\ee
This integral depends on the external matrix $\Lambda$, which is assumed to be a positive defined diagonal matrix. 
The times $t_k$ are given by the Miwa transform of the matrix $\Lambda$:
\be
t_{k}=\frac{1}{k}\Tr{\Lambda^{-k}}.
\ee 

After the shift of the integration variable
\be
\Phi=X-\Lambda
\ee
the numerator of (\ref{matint}) can be represented as
\be
e^{-\frac{1}{3}\Tr \Lambda^3}\int_{\mathcal H} \left[d X\right]\exp\left(-{\Tr\left(\frac{X^3}{3!}-\frac{\Lambda^2 X}{2}\right)}\right).
\ee
The Itzykson--Zuber integral (\ref{IZ2}) allows us to reduce the r.h.s. of (\ref{matint}) to the ratio of determinants
\be
\tau_{KW}\left(\left[Z\right]\right)=\frac{\det_{i,j=1}^N{\Phi^{KW}_i(z_j)}}{\Delta\left(z^{-1}\right)},
\ee
where $Z\equiv \Lambda^{-1}$ and the basis vectors are given by integrals
\be\label{KWbint}
\Phi^{KW}_k(z)=\frac{e^{-\frac{1}{3z^3}}}{\sqrt{2\pi z}}\int_{-\infty}^\infty   d y\,y^{k-1} \exp\left(-\frac{y^3}{3!}+\frac{y}{2z^2}\right)\\
=\frac{1}{\sqrt{2\pi z}}\int_{-\infty}^\infty   d y\,(y+z^{-1})^{k-1} \exp\left(-\frac{y^3}{3!}-\frac{y^2}{2z}\right). 
\ee
The coefficients of the basis vectors can be found explicitly, in particular
\be
\Phi^{KW}_1(z)=\sum_{k=0}^\infty\frac{2^k\,\Gamma\left(3k+\frac{1}{2}\right)}{9^k\,(2k)!\,\Gamma\left(\frac{1}{2}\right)}z^{3k},\\
\Phi^{KW}_2(z)=-\sum_{k=0}^\infty\frac{6k+1}{6k-1}\frac{2^k\,\Gamma\left(3k+\frac{1}{2}\right)}{9^k\,(2k)!\,\Gamma\left(\frac{1}{2}\right)}z^{3k-1}.
\ee

The first line of (\ref{KWbint}) allows us to find the Kac--Schwarz operators of the KW tau-function \cite{Kac,KS2}. Indeed, we have:
\be\label{recre}
\Phi^{KW}_{k+1}(z)=\frac{e^{-\frac{1}{3z^3}}}{\sqrt{2\pi z}} \left(-z^3\frac{\p}{\p z}\right)\int_{-\infty}^\infty   d y\,y^{k-1} \exp\left(-\frac{y^3}{3!}+\frac{y}{2z^2}\right)=a_{KW}\,\Phi^{KW}_{k}(z), 
\ee
where
\be\label{CS1}
a_{KW}=\frac{1}{z}-z^3 \frac{\p}{\p z} -\frac{z^2}{2}.
\ee
Thus,
\be
a_{KW} \left\{\Phi^{KW}\right\} \subset \left\{\Phi^{KW}\right\}
\ee
and the operator $a_{KW}$ is the Kac--Schwarz operator.

To construct another Kac--Schwarz operator we use the identity\footnote{Corresponding operator generates the $\mathcal{D}$-module describing the tau-function \cite{Dijk}.}
\be\label{KWqc}
\left(a_{KW}^2-z^{-2}\right)\Phi^{KW}_{1}(z)=0.
\ee
From this identity and the recursion relation (\ref{recre}) it follows that
\be\label{bKW}
z^{-2}\Phi_k^{KW}=\Phi_{k+2}^{KW}-2(k-1)\Phi_{k-1}^{KW}.
\ee
Thus, 
\be\label{CS2}
b_{KW}=z^{-2}
\ee
is also the Kac--Schwarz operator. Kac--Schwarz operators (\ref{CS1}) and (\ref{CS2}) satisfy the canonical commutation relation
\be\label{CanKW}
\left[a_{KW},b_{KW}\right]_-=2
\ee
and generate an algebra of the Kac--Schwarz operators for the KW tau-function.

The Kac--Schwarz operators that have been constructed allow us to find two infinite series of operators, which annihilate the tau-function. Let us consider the operators
\be\label{CSsim}
\widehat{J}_{k}^{KW}\equiv\widehat{Y}_{(b_{KW})^k}=\widehat{Y}_{z^{-2k}}=\frac{\p}{\p t_{2k}}, \,\,\,\,\,\mathrm{for} \quad k\geq 1.
\ee
where we have used (\ref{Yquant}). From the general properties of the Kac--Schwarz operators it follows that the KW tau-function is an eigenfunction of the operators (\ref{CSsim}). The same is true for the Virasoro operators 
\be
\widehat{L}_k^{KW}\equiv\widehat{Y}_{l_k^{KW}}+\frac{1}{16}\delta_{k,0}=\frac{1}{2}\widehat{L}_{2k}-\frac{1}{2}\frac{\p}{\p t_{2k+3}}+\frac{1}{16}\delta_{k,0},
\ee
where
\be\label{VirKS}
l_k^{KW}\equiv-\frac{1}{4}\left[(b_{KW})^{k+1},a_{KW}\right]_+=\frac{1}{2}\left( l_{2k}-z^{-2k-3}\right)
\ee
for $k\geq -1$.

To find corresponding eigenvalues it is enough to check that these operators satisfy the commutation relations of the algebra $\mathcal{V}_+$:
\be
\left[\widehat{J}_k^{KW},\widehat{J}_m^{KW}\right]_-=0,\,\,\,\, \mathrm{for} \,\,\,\, k,m\geq 1,\\
\left[\widehat{L}_{k}^{KW},\widehat{J}_m^{KW}\right]_-=-m\widehat{J}_{k+m}^{KW},\,\,\,\, \mathrm{for} \,\,\,\, k\geq -1 \,\,\,\, \mathrm{and} \,\,\,\, m\geq 1,\\
\left[\widehat{L}_{k}^{KW},\widehat{L}_{m}^{KW}\right]_-=(k-m)\widehat{L}_{k+m}^{KW},\,\,\,\, \mathrm{for} \,\,\,\, k,m\geq -1.
\ee
Since all generators of the algebra can be obtained as commutators of some other generators, the eigenvalues of all of them are equal to zero:
\be\label{KWfirst}
\widehat{J}_m^{KW} \tau_{KW} =0,~~~m\geq1
\ee
and
\be\label{KWvir}
\widehat{L}_{m}^{KW}\tau_{KW}=0,~~~m\geq-1.
\ee
Obviously, the first identity is just another way to say that the KW tau-function does not depend on even times (and is a tau-function of the KdV hierarchy).

Then, for any function $Z({\bf t^o})$ depending only on the odd times $t_{2m+1}$, we have 
\be
\widehat{\mathcal{L}}_k Z({\bf t^o})=\left(\widehat{L}_{2k}+\frac{1}{8}\delta_{k,0}\right) Z({\bf t^o}),~~~k\geq-1,
\ee
where the operators
\be
\widehat{\mathcal{L}}_m=\sum_{k=1}^\infty \left(2k+1\right)t_{2k+1}\frac{\p}{\p t_{2k+2m+1}}+\frac{1}{2}\sum_{k=0}^{m-1}\frac{\p^2}{\p t_{2k+1} \p t_{2m-2k-1}}+\frac{t_1^2}{2}\delta_{m,-1}+\frac{1}{8}\delta_{m,0},~~~m\geq -1
\label{VirK}
\ee
constitute the same subalgebra of the Virasoro algebra:
\be
\left[\widehat{\mathcal{L}}_n,\widehat{\mathcal{L}}_m\right]_-=2(n-m)\widehat{\mathcal{L}}_{n+m},~~~n,m\geq-1.
\ee
Thus, the Virasoro constraints from (\ref{KWvir}) are equivalent to the standard Virasoro constraints for the KW tau-function
\be
\widehat{\mathcal{L}}_m \tau_{KW}=\frac{\p}{\p t_{2m+3}}\tau_{KW},~~~m\geq-1.
\label{vir}
\ee
The Virasoro operators $\widehat{\mathcal{L}}_m$ are combinations of the even part of the current (\ref{boscur}):
\be
\widehat{\mathcal{J}}(z)=\frac{\widehat{J}(z)-\widehat{J}(-z)}{2}=\sum_{k=1}^\infty \left((2k+1) t_{2k+1} z^{2k} +\frac{1}{z^{2k+2}}\frac{\p}{\p t_{2k+1}}\right),
\ee
namely
 \be
\normordboson\widehat{\mathcal{J}}(z)^2\normordboson=2\sum_{k=-\infty}^\infty\frac{\widehat{\mathcal{L}}_k}{z^{2k+2}}-\frac{1}{4z^2}.
\ee




\subsection{Hurwitz tau-function}\label{HURWITZ}

According to \cite{HHK,HHK1} the Hurwitz tau-function, given by (\ref{Hurwitzgf}) can be represented in terms of the cut-and-join operator (\ref{CAJ}):
\be\label{cajtau}
\tau_H({\bf t};\beta)=\exp\left(\frac{ \beta}{2}\widehat{M}_0\right)\exp\left({t_1}\right)\\
=1+{t_1}+\frac{e^{\beta}}{2} \left({t_1^2}+{t_2}\right)+\frac{e^{-\beta }}{2} \left({t_1^2}-{t_2}\right)+\cdots
\ee
The coefficients of the parameter $\beta$ expansion are polynomials in $t_k$ (up to an exponential prefactor):
\be
\tau_H({\bf t};\beta)={{\rm e}^{t_{{1}}}}\left(1+t_{{2}}{\beta}+\frac{1}{4}\, \left( {t_{{1}}}^{2}+2\,{t_{{2}}}^{2}+6\,t_{{3}} \right) {\beta}^
{2}\right.\\
\left.+\frac{1}{12}\,\left( 3\,{t_{{1}}}^{2}t_{{2}}+2\,{t_{{2
}}}^{3}+18\,t_{{3}}t_{{2}}+16\,t_{{2}}t_{{1}}+2\,t_{{2}}+32\,t_{{4}}
 \right) {{\beta}}^{3}
+\cdots
\right).
\ee

From (\ref{cajtau}) we can easily derive an expression for the basis vectors \cite{Kazarian}: using the notations of Section {\ref{S1}} we get
\be
\Phi_k^H(z)= (\gamma_k^k)^{-1} \exp{\left(\frac{\beta}{2}m_0\right)} \exp(j_{-1})\, z^{1-k},
\ee
where $\gamma_{k}^{k}= \mbox{res}_{z} \left(z^{-k} \exp{\left(\frac{\beta}{2}m_0\right)}\, z^k\right)=\exp\left(\frac{\beta}{24}+\frac{\beta}{2}\left(k-\frac{1}{2}\right)^2\right) $, so that 
\be
\Phi_k^H(z)=e^{\frac{\beta}{2}\left(\left(z\p_z-\frac{1}{2}\right)^2-\left(k-\frac{1}{2}\right)^2\right)} e^z z^{1-k}
=\sum_{i=0}^\infty e^{\frac{\beta}{2}\left(\left(i-k+\frac{1}{2}\right)^2-\left(k-\frac{1}{2}\right)^2\right)}\frac{z^{i-k+1}}{i!}.  
\ee
This basis is not a canonical one. It is convenient to rewrite the basis vectors $\Phi_k^H$ in terms of integrals in one variable \cite{MMRP}. Namely, for any operator $A$ we have
\be
e^{\frac{\beta}{2}A^2}=\frac{1}{\sqrt{2\pi\beta}}\int_{-\infty}^\infty   d y\, \exp\left(-\frac{y^2}{2\beta}+y A\right), 
\ee
so that, since $m_0=\left(z\p_z-\frac{1}{2}\right)^2+\frac{1}{12}$, we have
\be\label{Hurbv}
\Phi_k^{H} (z) = \frac{\exp\left(-\frac{\beta}{2}\left(k-\frac{1}{2}\right)^2\right)}{\sqrt{2\pi\beta}}\int_{-\infty}^\infty   d y\,(z e^y)^{1-k} \exp\left(-\frac{y^2}{2\beta}-\frac{y}{2}+z e^y\right), 
\ee
or, combining terms in a different way,
\be\label{Hurbv2}
\Phi_k^{H} (z)=\frac{z^{1-k}}{\sqrt{2\pi\beta}}\int_{-\infty}^\infty   d y\,\exp\left(-\frac{1}{2\beta}\left(y+\left(k-\frac{1}{2}\right)\beta\right)^2+z e^y\right). 
\ee
This integral can diverge (for example, it diverges for any real nonzero $z$) and should be only considered as a formal Laurent series. The tau-function in the Miwa parametrization is
\be\label{Hureig}
\tau_{H}\left(\left[Z\right];\beta\right)=\frac{\det_{i,j=1}^N{\Phi_i^{H}(z_j)}}{\Delta\left(z^{-1}\right)}\\
=\frac{e^{\beta\left(\frac{N}{24}-\frac{N^3}{6}\right)}}{(2\pi\beta)^\frac{N}{2}\Delta({z}^{-1})} \prod_{i=1}^N \int_{-\infty}^\infty
dy_i\,  \Delta\left({z}^{-1} e^{-y}\right)\exp\left(-\sum_{i=1}^NW_H(z_i,y_i)\right)
\ee
with the potential
\be
W_H(z,y)=\frac{y^2}{2\beta}+\frac{y}{2}-{z}e^{y}.
\ee


This eigenvalue integral representation allows us to construct several matrix model representations. First of all, one can introduce additional integration variables, given by the unitary matrix. As a result, the Hurwitz tau-function can be represented as a Hermitian matrix integral \cite{Morsh,MMRP}:
\be
\tau_H\left(\left[Z\right];\beta\right)={\mathcal{P}}_H^{-1}\int_{{\mathcal H}} \left [ d \mu(\Phi) \right] \exp\left(-\frac{1}{2\beta}\Tr \Phi^2 +\Tr \left(e^{\Phi-N\beta/2}Z\right)\right),
\ee
where the measure of integration $\left [ d \mu(\Phi) \right]$ is non-flat and can be represented in terms of the standard measure on the space of Hermitian matrices, given by (\ref{hermes}), as follows:
\be
\left[ d \mu(\Phi) \right]=\sqrt{\det\frac{\sinh\left(\frac{{\Phi}\otimes {1}-{1}\otimes{\Phi}}{2}\right)}{\left(\frac{{\Phi}\otimes {1}-{1}\otimes{\Phi}}{2}\right)}} \left[ d \Phi \right].
\ee
The coefficient ${\mathcal{P}}_H=\int_{\mathcal H} \left [ d \mu(\Phi) \right] \exp\left(-\frac{1}{2\beta}\Tr \Phi^2 \right)$ does not depend on the external matrix $Z$.

The same tau-function can be given by an integral over normal matrices \cite{MMRP}. Namely, the integral (\ref{Uni}) in the eigenvalue form gives
\be
\tau_H^{N}({\bf t};\beta)=\int_\mathcal{U} \left[d U\right] \exp\left(\sum_{k=1}^\infty t_k U^{\dagger k}\right ) \tau_H\left(\left[U\right];\beta\right)\\
=\frac{1}{N!} \prod_{i=1}^N \left(\frac{1}{2\pi i}\oint \frac{d z_i}{z_i} \right)\Delta(z)
\Delta(z^{-1})\exp\left(\sum_{k=1}^\infty \sum_{j=1}^N t_k z_j^{-k}\right) \tau_{H}\left(\left[Z\right];\beta\right).
\ee
On substitution of (\ref{Hureig}) into this relation we get an integral of the form 
\be\label{IntH}
\prod_{i=1}^N \left(\int_{-\infty}^\infty d\,y_i\oint \frac{d z_i}{2\pi iz_i} \right)\Delta(z)\Delta(z^{-1}e^{-y}) \exp\left({\sum_{i=1}^N{W}(z_i,y_i)}\right)
\ee
for some known function ${W}(z,y)$. This integral can be reduced to a normal matrix integral with the measure 
(\ref{NMMm}). Indeed, let us assume that $z_i$ integral is taken over the circle $|z_i|=\exp(-y_i/2)$, so that $|z_i^{-1}|=\exp(-y_i/2)$. Then
\be
(\ref{IntH})= \pi^{-N}\prod_{i=1}^N \int_{\raise-1pt\hbox{$\mbox{\Bbbb C}$}}\frac{d^2 z_i}{|z_i|^2} \,\Delta(z) \Delta(\bar{z}) \exp\left({\sum_{i=1}^N{W}(z_i,-\log|z_i|^2)}\right),
\ee
where $d^2z=d \left(\Re\, z\right)d \left(\Im\, z\right)$, so that the tau-function is given by a normal matrix integral
\be
\tau_H^N({\bf t};\beta)= {\mathcal P}^{-1}\int_{\mathcal{N}} \left[d Z\right]
\exp\left(- \Tr \widetilde{W}_H+\sum_{k=1}^\infty t_k  \Tr { Z}^{-k}\right),
\ee
where 
\be
\widetilde{W}_H={\frac{1}{2\beta}\left(\log {{Z^\dagger Z }}\right)^2}+\frac{1}{2} \log {{Z Z^\dagger}}- (Z^\dagger)^{-1}
\ee
and $ {\mathcal P}$, again, does not depend on times. 

\subsection{Kac--Schwarz description of the Hurwitz tau-function}\label{KSHH}
Let us find the Kac--Schwarz operators for the Hurwitz tau-function. First of all, from (\ref{Hurbv2}) we see that
\be\label{KSH1}
\exp\left(\pm \beta z \frac{\p}{\p z}\right) \Phi_k^H(z)\\
= \frac{\left(ze^{\pm \beta}\right)^{1-k}}{\sqrt{2\pi\beta}}\int_{-\infty}^\infty   d y\,\exp\left(-\frac{1}{2\beta}\left(y+\left(k-\frac{1}{2}\right)\beta\right)^2+z e^{y\pm \beta}\right)\\
= 
\frac{\left(ze^{\pm \beta}\right)^{1-k}}{\sqrt{2\pi\beta}}\int_{-\infty}^\infty   d y\,\exp\left(-\frac{1}{2\beta}\left(y+\left(k\mp 1-\frac{1}{2}\right)\beta\right)^2+z e^{y}\right),
\ee
where we have shifted the integration variable $y \mapsto y\mp\beta$. Therefore, operator
\be\label{bhurw}
b_H=z^{-1}\exp\left(-\beta z \frac{\p}{\p z}\right)
\ee
is the Kac--Schwarz operator
\be\label{bhact}
b_H\, \Phi_k^H=e^{\beta(k-1)}\Phi_{k+1}^H.
\ee
Combining relation (\ref{KSH1}) with an identity
\begin{multline}\label{OPh}
z \frac{\p}{\p z}\, \Phi_k^H(z)= \frac{\exp\left(-\frac{\beta}{2}\left(k-\frac{1}{2}\right)^2\right)}{\sqrt{2\pi\beta}}\int_{-\infty}^\infty   d y\,\left(1-k+z e^y\right) \\
\times(z e^y)^{1-k} \exp\left(-\frac{y^2}{2\beta}-\frac{y}{2}+z e^y\right)
\end{multline}
we get 
\be\label{KSHuw}
\left(z \exp\left(\beta z \frac{\p}{\p z}\right) - z \frac{\p}{\p z}\right) \Phi_k^H =(k-1)\,\Phi_k^H.
\ee
Thus, the operator
\be\label{Huraop}
a_H=z \exp\left(\beta z \frac{\p}{\p z}\right) - z \frac{\p}{\p z}\\
\ee
is the Kac--Schwarz operators for the Hurwitz tau-function.\footnote{The operators $a_H$ and $b_H$ can be obtained by a conjugation of the Kac--Schwarz operators $z-z\p_z$ and $z^{-1}$ for the tau-function $\exp(t_1)$
\be
a_H=e^{\frac{\beta}{2}m_0}\left(z-z\p_z\right)e^{-\frac{\beta}{2}m_0},\\
b_H=e^{-\beta} e^{\frac{\beta}{2}m_0}\,z^{-1}\,e^{-\frac{\beta}{2}m_0}.
\ee
 } 
 It is easy to check that these Kac--Schwarz operators satisfy the commutation relation
\be\label{comH}
\left[a_H,b_H\right]_-=b_H.
\ee

The operators $a_H$ and $b_H$ completely define a point of the Grassmannian. Indeed, on substitution of the first basis vector 
\be
\Phi_1^{H}= 1+ \sum_{m=1}\varphi_{m,1}^H z^m
\ee
into (\ref{KSHuw}) for $k=1$ we obtain the first coefficient, $\varphi_{1,1}^H=1$, and a recursive relation, which allows us to find all other $\varphi_{m,1}^H$:
\be
\varphi_{m+1,1}^H=\frac{e^{m\beta}}{m+1}\varphi_{m,1}^H.
\ee
This yields $\Phi_1^{H}$, while all higher basis vectors can be obtained with the help of
(\ref{bhact}).

The operators $a_H$ and $b_H$ can be quantized via (\ref{Yquant}), and for the obtained operators $\widehat{Y}_{a_H}$ and $\widehat{Y}_{b_{H}}$ the Hurwitz tau-function is an eigenfunction
\be\label{KSH}
\widehat{Y}_{a_H}=t_1-\widehat{L}_0+\beta \widehat{L}_{-1}+\frac{\beta^2}{2}\widehat{M}_{-1}+\ldots,\\
\widehat{Y}_{b_{H}}=\frac{\p}{\p t_1}-\beta\left(\widehat{L}_1+\frac{\p}{\p t_1}\right)+
\frac{\beta^2}{2}\left(\widehat{M}_{1}+2\widehat{L}_1+\frac{\p}{\p t_1}\right)+\ldots.
\ee
It is easy to find the corresponding eigenvalue for the operator $\widehat{Y}_{a_H}$.
Indeed, it is obvious that the operator $\hat{Y}_{z \exp\left(\beta z \p_z \right)}$ has positive energy. Thus, $\widehat{Y}_{z \exp\left(\beta z \p_z \right)} Z({\bf t})|_{{\bf t}={\bf 0}}=0$ for all functions $Z({\bf t})$, and we have
\be
\widehat{Y}_{a_H} \tau_H({\bf t})|_{{\bf t}={\bf 0}}=-\widehat{Y}_{z \p_z}\tau_H({\bf t})|_{{\bf t}={\bf 0}}=0,
\ee
so that
\be
\widehat{Y}_{a_H} \tau_H({\bf t})=0.
\ee

To find the eigenvalue of the operator $\widehat{Y}_{b_H}$ let us consider a commutator of the operators (\ref{KSH}). It is a deformation of the commutation relation (\ref{comH}): 
\be
\left[\widehat{Y}_{a_H},\widehat{Y}_{b_H}\right]_-=\widehat{Y}_{b_H}+C_H,
\ee
where $C_H$ is a constant. Since the Hurwitz tau-function is an eigenfunction for both operators, for the commutator we have
\be
\left[\widehat{Y}_{a_H},\widehat{Y}_{b_H}\right]_- \tau_H({\bf t})=0, 
\ee
so that
\be
\widehat{Y}_{b_H}\, \tau_H({\bf t})=-C_H\, \tau_H({\bf t}).
\ee
From the definition of operators ${Y}_{a}$ and ${R}_k$ in (\ref{Rdif}) and (\ref{Ydif}) we have
\be\label{Yphicor}
\widehat{Y}_{z^{-k}q^{z\p_z}}=-q^{k+1}\widehat{R}_k(q).
\ee
Thus,
\be\label{Yah}
\widehat{Y}_{a_H}=-\widehat{\Phi}_{-1}(e^{\beta})-\widehat{L}_0,\\
\widehat{Y}_{b_H}=-e^{-2\beta}\widehat{\Phi}_{1}(e^{-\beta}),
\ee
so that the commutation relation (\ref{comcoin}) yields
\be
\left[\widehat{Y}_{a_H},\widehat{Y}_{b_H}\right]_-=\left[\widehat{\Phi}_{-1}(e^{\beta})+\widehat{L}_0,e^{-2\beta}\widehat{\Phi}_{1}(e^{-\beta})\right]_-=\widehat{Y}_{b_H}-e^{-\beta}.
\ee
 Thus, $C_H=-e^{-\beta}$ and
 \be
 \widehat{Y}_{b_H}\, \tau_H({\bf t})=e^{-\beta}\, \tau_H({\bf t}).
 \ee

Let us consider two families of operators
\be
j_k^{H}=\left(b_H\right)^{k} \,\,\,\, \mathrm{for} \,\,\, k\geq 1,\\
l_k^{H}=-\left(b_H\right)^ka_H\,\,\,\,\,\mathrm{for} \,\,\, k\geq -1.
\ee
It is obvious that all these operators (maybe except for $l_{-1}^H$) belong to the Kac--Schwarz algebra. For $l_{-1}^H$ this statement can be checked explicitly:
\be
\l_{-1}^{H}\Phi_k^H=(1-k)e^{(2-k)\beta}\Phi_{k-1}.
\ee
Our notations reflect the fact that these operators are the deformations of the standard operators (\ref{obdif}):
\be\label{expH}
j_k^{H}=j_k+O(\beta),\\
l_k^{H}=l_k+\frac{k+1}{2}j_k-j_{k-1}+O(\beta).
\ee
where we assume that $j_0=1$. From the commutation relation (\ref{comH}) one can easily derive the commutation relations
\be
\left[j_k^{H},j_m^{H}\right]_-=0,\\
\left[l_k^{H},j_m^{H}\right]_-=-m\,j_{k+m}^H,\\
\left[l_k^{H},l_m^{H}\right]_-=(k-m)l_{k+m}^H,
\ee
which coincide with the commutation relations of the subalgebra $\mathcal{V}_+$ of the Heisenberg--Virasoro algebra. We have the following explicit expressions: 
\be
j_k^{H}=e^{\frac{k^2-k}{2}\beta}z^{-k}\exp\left(-k \beta z \frac{\p}{\p z}\right),\\
l_k^H=-e^{\frac{k^2-3k}{2}\beta}z^{1-k}\exp\left((1-k) \beta z \frac{\p}{\p z}\right)+e^{\frac{k^2-k}{2}\beta}z^{-k}\exp\left(-k \beta z \frac{\p}{\p z}\right)z\frac{\p}{\p z}.
\ee
From (\ref{Yphicor}) for $k>0$ we have
\be
\widehat{Y}_{z^{-k}e^{-k\beta z\p_z}}=-e^{-(k+1)k\beta}\,\widehat{\Phi}_k\left(e^{-k\beta}\right),\\
\widehat{Y}_{z^{-k}e^{-k\beta z\p_z}z\p_z}=\frac{1}{k}\frac{\p}{\p \beta}\left(e^{-(k+1)k\beta}\,\widehat{\Phi}_k\left(e^{-k\beta}\right)\right).
\ee
Thus,
\be
\widehat{Y}_{j^H_k}=-e^{-\frac{k^2+3k}{2}\beta}\widehat{\Phi}_k\left(e^{-k\beta}\right),
\ee
and for $\widehat{Y}_{l^H_k}$ we have
\be
\widehat{Y}_{l^H_0}=\widehat{L}_0+\widehat{\Phi}_{-1}(e^{\beta})=-\widehat{Y}_{a_H},\\
\widehat{Y}_{l^H_1}=\frac{\p}{\p \beta}e^{-2\beta}\widehat{\Phi}_{1}(e^{-\beta})=-\frac{\p}{\p \beta}\,\widehat{Y}_{b_H},\\
\widehat{Y}_{l^H_k}=e^{-\frac{k^2+k}{2}\beta}\widehat{\Phi}_{k-1}\left(e^{(1-k)\beta}\right)\\
+\frac{1}{k}e^{\frac{k^2-k}{2}\beta}\frac{\p}{\p \beta}\left(e^{-(k+1)k\beta}\,\widehat{\Phi}_k\left(e^{-k\beta}\right)\right),\,\,\,\, \mathrm{for}\,\,\,\,\, k=-1\,\,\,\,\, \mathrm{and} \,\,\,\,\,k>1.
\ee
An expansion of these operators follows for the expansion (\ref{expH}) 
\be
\widehat{Y}_{j^H_k}=\widehat{J}_k+O(\beta),\\
\widehat{Y}_{l^H_k}=\widehat{L}_k+\frac{k+1}{2}\widehat{J}_k-\widehat{J}_{k-1}+O(\beta).
\ee
The relations (\ref{Phicom}), (\ref{comcoin}) and the commutation relation
\be
\left[\widehat{L}_0,\widehat{\Phi}_{k}(q)\right]_-=-k\,\widehat{\Phi}_{k}(q)
\ee
allow us to find the commutation relations between the operators $\widehat{Y}_{j^H_k}$ and $\widehat{Y}_{l^H_k}$
\be
\left[\widehat{Y}_{j^H_k},\widehat{Y}_{j^H_m}\right]_-=0,\\
\left[\widehat{Y}_{l^H_k},\widehat{Y}_{j^H_m}\right]_-=-m\left(\widehat{Y}_{j^H_{k+m}}-\delta_{m+k,1}e^{-\beta}\right),\\
\left[\widehat{Y}_{l^H_k},\widehat{Y}_{l^H_m}\right]_-=(k-m)\left(\widehat{Y}_{l^H_{k+m}}-e^{-\beta}\delta_{k+m,1}+\frac{1}{2}e^{-2\beta}\delta_{k+m,2}\right).
\ee
Thus, the operators
\be
\widehat{J}_k^{H}=\widehat{Y}_{j^H_{k}}-\delta_{k,1}e^{-\beta},\,\,\,\,k\geq1\\
\widehat{L}_{k}^{H}=\widehat{Y}_{l^H_{k}}-e^{-\beta}\delta_{k,1}+\frac{1}{2}e^{-2\beta}\delta_{k,2},\,\,\,k\geq -1.
\ee
satisfy the commutation relations of the algebra $\mathcal{V}_+$:
\be
\left[\widehat{J}_k^{H},\widehat{J}_m^{H}\right]_-=0,\,\,\,\, \mathrm{for}\,\,\,\, k,m\geq 1,\\
\left[\widehat{L}_{k}^{H},\widehat{J}_m^{H}\right]_-=-m\widehat{J}_{k+m}^{H},\,\,\,\, \mathrm{for} \,\,\,\, k\geq -1 \,\,\,\, \mathrm{and} \,\,\,\, m\geq 1, \\
\left[\widehat{L}_{k}^{H},\widehat{L}_{m}^{H}\right]_-=(k-m)\widehat{L}_{k+m}^{H}, \,\,\,\, \mathrm{for}\,\,\,\, k,m\geq -1.
\ee
Since $\tau_H$ is an eigenfunction of the operators $\widehat{J}_k^{H}$ and $\widehat{L}_{k}^{H}$, from these commutation relations it follows that the Hurwitz tau-function satisfies the constraints
\be
\widehat{J}_m^{H} \tau_H =0, \,\,\,\, \mathrm{for}\,\,\,\, m\geq 1,\\
\widehat{L}_{m}^{H}\tau_H=0, \,\,\,\, \mathrm{for}\,\,\,\, m\geq -1.
\ee

\subsection{Hodge tau-function}

The Hodge tau-function (\ref{gf}) is a deformation of the KW tau-function:
\be\label{Hodgeex}
\tau_{Hodge}({\bf t};u)=\tau_{KW}({\bf t})\\
+\left(\frac{1}{6}\,t_{{2}}+  \frac{5}{3}\,t_{{4}}t_{{1}}+{\frac {25}{48}}\,t_{{
3}}t_{{2}}+{\frac {25}{36}}\,t_{{2}}{t_{{1}}}^{3} +
 {\frac {35}{8}}\,t_{{8}}+{\frac {1225}{768}}\,t_{{2}}{t_{{3}}}
^{2}+{\frac {35}{18}}\,{t_{{1}}}^{4}t_{{4}}+{\frac {35}{4}}\,{t_{{1}}}
^{2}t_{{6}}\right.\\
\left.+{\frac {245}{48}}\,t_{{5}}t_{{1}}t_{{2}}+{\frac {49}{432}}
\,{t_{{1}}}^{6}t_{{2}}+{\frac {1225}{288}}\,t_{{2}}{t_{{1}}}^{3}t_{{3}
}+{\frac {245}{24}}\,t_{{3}}t_{{1}}t_{{4}} +\cdots
\right)u\\
+\left(  {\frac {25}{72}}\,{t_{{2}}}^{2}+\frac{1}{6}\,{t_{{1}}}^{4}+\frac{3}{2}\,t_{{3}
}t_{{1}} +  {\frac {2737}{288}}\,t_{{7}}+{\frac {
245}{18}}\,t_{{4}}t_{{1}}t_{{2}}+{\frac {133}{48}}\,t_{{3}}{t_{{1}}}^{
4}+{\frac {175}{12}}\,t_{{5}}{t_{{1}}}^{2}+{\frac {1225}{432}}\,{t_{{2
}}}^{2}{t_{{1}}}^{3}\right.\\
\left.+\frac{1}{36}\,{t_{{1}}}^{7}+{\frac {147}{16}}\,{t_{{3}}}^
{2}t_{{1}}+{\frac {1225}{576}}\,t_{{3}}{t_{{2}}}^{2} +\cdots
\right)u^2\\
+\left(\frac{1}{2}\,t_{{2}}t_{{1}}+ {\frac {1225}{1296}}\,{t_{{2}}}^{3}
+{\frac {343}{32}}\,t_{{6}}+{\frac {245}{16}}\,t_{{3}}t_{{1}}t_{{2}}+{
\frac {77}{6}}\,{t_{{1}}}^{2}t_{{4}}+{\frac {35}{18}}\,{t_{{1}}}^{4}t_
{{2}}+\cdots 
\right)u^3+\cdots.
\ee
It is related to the Hurwitz tau-function by (\ref{Kaz}), or, in other words, by a linear change of variables accomplished by cancelation of proper linear and quadratic terms in free energy. Let us describe this relation in more detail.

The prefactor  $\exp\left({-\sum_{k>0}\frac{k^{k-1}\beta^{k-1}t_k}{k!}}\right)$ in (\ref{HtoH}) is responsible for genus zero one-point contributions from the non-stable maps.
On the level of basis vectors it corresponds to the multiplication:
\be
\Phi^{(0)}_k(z)=\exp\left( -\sum_{k=1}^\infty\frac{k^{k-2}}{k!} \beta^{k-1} z^k\right)\Phi_k^H(z).
\ee

The next step in the chain of transformations (\ref{Kaz}) is given by the operator $\widehat{L}_-^{(1)}$, which, according to (\ref{viract}), on the level of basis vectors yields 
\begin{equation}\label{Phi1}
\begin{split}
\Phi_k^{(1)}(z)&
=\frac{\beta z}{f_{-1}(\beta z)}\sqrt{\frac{1}{\beta}\frac{\p f_{-1}(\beta z)}{\p z}}\, \Phi_k^{(0)}\left(\frac{f_{-1}(\beta z)}{\beta}\right)\\
&=\sqrt{1-\beta z}\,e^{\frac{\beta z(1+z)}{2}-z}\sum_{i=0}^\infty e^{\frac{\beta}{2}\left(\left(i-k+\frac{1}{2}\right)^2-\left(k-\frac{1}{2}\right)^2\right)}\frac{\left(z e^{-\beta z}\right)^{i-k+1}}{i!} \\  
&=\sqrt{\frac{{1-\beta z}}{{2\pi\beta}}}e^{-\frac{\beta}{2}\left(k-\frac{1}{2}\right)^2+\frac{\beta z(1+z)}{2}-z}\int_{-\infty}^\infty   d y\,(z e^{y-\beta z})^{1-k} e^{-\frac{y^2}{2\beta}-\frac{y}{2}+z e^{y-\beta z}},
\end{split}
\end{equation}
where $f_{-1}(z)$ is the series (\ref{minfunc}). After a change of integration variable $y \mapsto y+\beta z$ these basis vectors can be represented as
\be\label{Int1}
\Phi_k^{(1)}(z)=\sqrt{\frac{{1-\beta z}}{{2\pi\beta}}}e^{-\frac{\beta}{2}\left(k-\frac{1}{2}\right)^2}\int_{-\infty}^\infty   d y\,(z e^{y})^{1-k} \exp\left(-\frac{y^2}{2\beta}-\frac{y}{2}+z \left(e^{y}-1-y\right)\right). 
\ee
Corresponding tau-function in the Miwa parametrization is
\be
\tau_{(1)}\left(\left[Z\right]\right)=\frac{\det_{i,j=1}^N{\Phi_i^{(1)}(z_j)}}{\Delta\left(z^{-1}\right)}\\
=\frac{e^{\beta\left(\frac{N}{24}-\frac{N^3}{6}\right)}}{(2\pi\beta)^\frac{N}{2}\Delta({z}^{-1})} \prod_{i=1}^N \sqrt{1+\beta {z}_i}\int_{-\infty}^\infty
dy_i\,  \Delta\left({z}^{-1} e^{-y}\right)\exp\left(-\sum_{i=1}^N W_{(1)}(z_i,y_i)\right),
\ee
where the potential is
\be
W_{(1)}(z,y)=\frac{y^2}{2\beta}+\frac{y}{2}-{z} \left(e^{y}-y-1\right).
\ee
Independent time variables can be introduced with the help of (\ref{Uni}):
\begin{multline}\label{minto}
\tau^N_{(1)}({\bf t})= \frac{1}{N!} \prod_{i=1}^N \left(\frac{1}{2\pi i}\oint \frac{d z_i}{z_i} \right)\Delta(z)
\Delta(z^{-1})\exp\left(\sum_{k=1}^\infty \sum_{j=1}^N t_k z_j^{-k}\right) \tau_{(1)}\left(\left[Z\right]\right)
\\
=\frac{e^{\beta\left(\frac{N}{24}-\frac{N^3}{6}\right)}}{(2\pi\beta)^\frac{N}{2}N! } \prod_{i=1}^N \left(\frac{1}{2\pi i}\oint \frac{d z_i}{z_i} \int_{-\infty}^\infty d y_i\right)\Delta(z) \Delta\left({z}^{-1} e^{-y}\right)\\
\times\det(1-\beta Z)^{\frac{1}{2}}\prod_{i=1}^N e^{-W_{(1)}(z_i,y_i)+\sum_{k=1}^\infty  t_k z_i^{-k}}.
\end{multline}
Using the change of variables from the previous section we reduce the tau-function to the normal matrix integral
\be
\tau^N_{(1)}({\bf t})={\mathcal P}^{-1}_{(1)}\int_{\mathcal N} \left[d Z\right]
\exp\left( -\Tr \widetilde{W}_{(1)}+\sum_{k=1}^\infty t_k \Tr { Z}^{-k}\right),
\ee
where the normalization factor ${\mathcal P}_{(1)}=\int_{\mathcal N} \left[d Z\right]
\exp\left( -\Tr \widetilde{W}_{1}\right)$ does not depend on the variables $t_k$
and
\be\label{1pot}
\widetilde{W}_{(1)}=\frac{1}{2\beta}\left(\log {Z^\dagger Z }\right)^2+\frac{1}{2} \log {{Z^\dagger Z }}-(Z^\dagger)^{-1}+Z -Z \log {{Z^\dagger Z }}-\frac{1}{2}\log\left(1-\beta Z\right).
\ee
Let us check this matrix integral formula for $N=1$ and the Miwa parametrization of times $t_k=\frac{1}{k} y^k$. From (\ref{minto}) we have
\be
\tau_{(1)}(\left[y\right])=\frac{1}{\pi}\frac{e^{-\frac{\beta}{8}}}{\sqrt{2\pi\beta}} \int_{\raise-1pt\hbox{$\mbox{\Bbbb C}$}}d^2 z
\frac{\sqrt{1-\beta z}}{1-y/z}e^{-\frac{1}{2\beta}\left(\log |z|^2\right)^2-\frac{1}{2}\log|z|^2+\bar{z}^{-1}-z+z\log|z|^2}\\
=1+\beta\,y+ \frac{1}{2}\left({
{\rm e}^{\beta}} -1-\beta+\frac{1}{2}\,{\beta}^{2} \right) {y}^{2}+ \frac{1}{2}\left( \frac{1}{3}\,
{{\rm e}^{3\,\beta}}-{{\rm e}^{\beta}}+\frac{1}{6}\,{\beta}^{3}-{{\rm e}^{\beta}}\beta+\beta+{\frac {2}{3}} \right) {y}^{3}\\
+ \left( {\frac {5}{24}}\,{\beta
}^{3}-{\frac {1}{8}}+\frac{3}{4}\,{{\rm e}^{\beta}}\beta-\frac{1}{3}\,{{\rm e}^{3\,
\beta}}\beta-\frac{1}{6}\,{{\rm e}^{3\,\beta}}+\frac{1}{8}\,{{\rm e}^{\beta}}{\beta}^{
2}+\frac{1}{24}\,{{\rm e}^{6\,\beta}}-{\frac {1}{96}}\,{\beta}^{4}+\frac{1}{4}\,{
{\rm e}^{\beta}}-{\frac {5}{12}}\,\beta \right) {y}^{4}\\
+O(y^5).
\ee
It is easy to check that this expression coincides with an expansion of the first basis vector $\Phi_1^{(1)}$ from (\ref{Phi1}) as it should be according to (\ref{rest2}). Other basis vectors can also be obtained in the form of integrals over $\mathbb{C}$.

The next step in the transformation from $\tau_{H}$ to $\tau_{Hodge}$ is given by the operator  $\widehat{L}_-^{(2)}$, corresponding to the series $f_{-2}(z)$ from (\ref{minfunc}). The operator $\widehat{L}_-^{(2)}$ is very simple:
\be
\widehat{L}_-^{(2)}=-\beta \widehat{L}_{-1}.
\ee
Indeed, consider an operator $l_{-1}=z^2 \frac{\p}{\p z}$ which shifts the variable $z^{-1}$. Then
\be
\exp\left(-\beta z^2 \frac{\p}{\p z}\right) \left[z\right]=\exp\left(\beta \frac{\p}{\p z^{-1}}\right) \left[\frac{1}{z^{-1}}\right]
=\frac{1}{z^{-1}+\beta}=\frac{z}{1+\beta z}=f_{-2}(z).
\ee
Since for a constant $a$ the conjugation with the operator $e^{a \widehat{L}_{-1}}$ yields
\be
e^{a \widehat{L}_{-1}}\, \left(\sum_{k>0} t_k z^k\right) \,e^{-a \widehat{L}_{-1}} = \sum_{k>0} t_k z^k+ a \sum_{k>0} k\, t_k z^{k-1}+\frac{a^2}{2}\sum_{k>0} k(k-1)\, t_k z^{k-2}+\dots\\
=\sum_{k>0} t_k \left((z+a)^k-a^k\right),
\ee
we can construct a matrix model, which describes the Hodge tau-function, when $N$ tends to infinity:\footnote{This can also be considered as a consequence of (\ref{dtr4}).}
\be\label{MMHodge1}
\widetilde{\tau}_{Hodge}^N\equiv\beta^{-\frac{4}{3}\widehat{L}_0} e^{\widehat{L}_{-}^{(2)}}\tau^N_{(1)}({\bf t})\\
={\mathcal P}^{-1}\int_{N\times N} \left[d Z\right]
\exp\left( -\Tr \widetilde{W}_{(1)}+\sum_{k=1}^\infty \beta^{-\frac{4}{3}k}t_k\, \Tr \left(\left({ Z}^{-1}-\beta\right)^k-(-\beta)^k\right)\right),
\ee
where the potential $\widetilde{W}_{(1)}$ is given by (\ref{1pot}).

One can take another route and obtain the basis vectors for the Hodge tau-function (they were obtained in \cite{Kazarian}) directly from the Hurwitz basis vectors $\Phi_k^H$:

\be\label{bvhodge}
\Phi^{Hodge}_k(z)=u^{4(1-k)} \frac{u^{-1} z}{f_-(u^{-1} z)}\sqrt{f_-'(u^{-1} z)}\,e^{-\frac{2zu+z^2}{2(u+z)^2u^3}} \,\Phi_k^{H} \left(u^{-3}f_-(u^{-1} z)\right) \\
=\frac{u^{4(1-k)}}{\sqrt{1+u^{-1} z}}e^{\frac{ z}{2(u+ z)}-\frac{2zu+z^2}{2(u+ z)^2u^3}}\,\sum_{i=0}^\infty e^{\frac{u^3}{2}\left(\left(i-k+\frac{1}{2}\right)^2-\left(k-\frac{1}{2}\right)^2\right)}\frac{\left(\frac{z}{(u+z)u^3}e^{-\frac{z}{u+z}}\right)^{i-k+1}}{i!}, 
\ee
or as a transformation of the vectors $\Phi_k^{(1)}$ and their integral representation (\ref{Int1}):
\be\label{Hodin}
\Phi^{Hodge}_k(z)=\beta^{-\frac{4}{3}(z\frac{\p}{\p z}+k-1)} e^{-\beta z^2 \frac{\p}{\p z}}\, \Phi_k^{(1)}(z)\\
=\frac{\left(z^{-1}+u^{-1}\right)^{k-1}}{\sqrt{2\pi u^3(1+zu^{-1})}}
\int_{-\infty}^\infty   d y\, \exp{\left(-\frac{1}{2u^3}\left(y+\left(k-\frac{1}{2}\right)u^3\right)^2+\frac{z}{(u+z)u^3} \left(e^{y}-1-y\right)\right)}.
\ee

The change of integration variable $y \mapsto y+ \log(1+u^{-1}z)$ yields
\be
\Phi^{Hodge}_k(z)=\frac{e^{-\frac{u^3}{2}\left(k-\frac{1}{2}\right)^2}}{\sqrt{2\pi u^3}(1+zu^{-1})}
\int_{-\infty}^\infty   d y\, \left(z\,e^{y}\right)^{1-k}\exp (-w_{Hodge}),
\ee
where
\be
w_{Hodge}=\frac{1}{2u^3}\left(y+\log(1+u^{-1}z)\right)^2-\frac{z}{u^4}e^y+\frac{z}{(u+z)u^3} \left(1+y+\log(1+u^{-1}z)\right).
\ee
Thus we can again introduce times with the help of the unitary matrix integral
\be\label{MMHodge2}
\tau_{Hodge}^N({\bf t},u)={\mathcal P}_{Hodge}^{-1}\int_{\mathcal N} \left[d Z\right]
\exp\left( -\Tr W_{Hodge}+\sum_{k=1}^\infty t_k \Tr { Z}^{-k}\right),
\ee
where
\be
W_{Hodge}=\frac{1}{2u^3} \left(\log\left(\frac{Z^{\dagger} Z}{1+u^{-1}Z}\right)\right)^2+\frac{1}{2}\log Z^\dagger Z+\log\left(1+u^{-1}Z\right)-\frac{1}{u^4 Z^\dagger}\\
+\frac{Z}{(Z+u)u^3}\left(1-\log\left(\frac{Z^{\dagger} Z}{1+u^{-1}Z}\right)\right).
\ee
In particular, for the Miwa parametrization with $N=1$ we get 
\be
\tau_{Hodge}(\left[y\right],u)=\frac{1}{\pi}\frac{e^{-\frac{u^3}{8}}}{\sqrt{2\pi u^3}} \int_{\raise-1pt\hbox{$\mbox{\Bbbb C}$}}d^2 z \frac{e^{-\frac{1}{2u^3} \left(\log\left(\frac{u\,|z|^2}{u+z}\right)\right)^2+\frac{1}{2}\log|z|^2+\frac{1}{u^4\bar{z}}-\frac{z}{(u+z)u^3}\left(1-\log\left(\frac{u\,|z|^2}{u+z}\right)\right)}}{(1-z^{-1}y)(1+u^{-1}z)}\\
=1+{\frac {2\,{{\rm e}^{{u}^{3}}}-2-2\,{u}^{3}-{u}^{6}}{4\,{u}^{8}}
}{y}^{2}+{\frac {2+{{\rm e}^{3\,{u}^{3}}}+12\,{u}^{3}+9\,{u}^{6}-
3\,{{\rm e}^{{u}^{3}}}+2\,{u}^{9}-12\,{{\rm e}^{{u}^{3}}}{u}^{3}}{6\,{u}^
{12}}}{y}^{3}+O \left( {y}^{4} \right), 
\ee
which coincides with the first basis vector $\Phi_1^{Hodge}$ from (\ref{bvhodge}). Let us stress that it is not obvious that the matrix models (\ref{MMHodge1}) and (\ref{MMHodge2}) (which do not coincide with each other for finite $N$) have finite limits, when $u$ tends to zero, while the corresponding tau-function in this limit reduces to the KW tau-function.

In what follows we need the properties of the series $\Phi^{Hodge}_k(z)$.
Namely, let us show that
\begin{lemma}\label{SING}
\be\label{sing}
\Phi_k^{Hodge}(z)=\left(\frac{1}{z}+\frac{1}{u}\right)^{k-1} \times (\mbox{formal Taylor series both in $z$ and $u$}).
\ee
\end{lemma}
\begin{proof}
The statement of the lemma follows from the integral representation (\ref{Hodin}). Indeed, let us change the variable of integration $y \mapsto y\, u\, \sqrt{u+z}$. Then
\be
\Phi_k^{Hodge}(z)=\frac{\left(z^{-1}+u^{-1}\right)^{k-1}}{\sqrt{2\pi }e^{\left(k-\frac{1}{2}\right)^2\frac{u^3}{2}}
}
\int_{-\infty}^\infty   d y\, e^{-\frac{1}{2}y^2-\left(k-\frac{1}{2}\right)\, \sqrt{u+z}\,u\, y
+z \sum_{k=3}^\infty(u+z)^{\frac{k}{2}-1} u^{k-3} \frac{y^k}{k!} }
\ee
In the expansion of the integral about the Gaussian potential only even powers of $y$ survive.
The coefficient in front of each of them is a polynomial in $u$ and $z$. Thus, the only source of negative powers of $z$ and $u$ is the prefactor, which completes the proof.
\end{proof}

Thus, among all basis vectors $\Phi^{Hodge}_k$ only the first one has a finite limit at $u=0$:
 \be
\Phi^{Hodge}_1(z)\Big|_{u=0}=\frac{1}{\sqrt{2\pi z}} \int_{-\infty}^\infty   d y \exp\left(-\frac{y^2}{2z}+\frac{y^3}{3!}\right)=\Phi^{KW}_1(z).
\ee
However, it is easy to construct another admissible basis, regular at $u=0$:
\be
\widetilde{\Phi}_1^{Hodge}=\Phi_1^{Hodge},\\
\widetilde{\Phi}_2^{Hodge}=\Phi_2^{Hodge}-u^{-1}\Phi_1^{Hodge},\\
\widetilde{\Phi}_3^{Hodge}=\Phi_3^{Hodge}-2u^{-1}\Phi_2^{Hodge}+u^{-2}\Phi_1^{Hodge},\\
\widetilde{\Phi}_4^{Hodge}=\Phi_3^{Hodge}-3u^{-1}\Phi_3^{Hodge}+3u^{-2}\Phi_2^{Hodge}-u^{-3}\Phi_1^{Hodge},\\
\dots\\
\ee

\subsection{Simplified relation}\label{Sim}
In this section, we simplify the relation (\ref{FromKWtoH}).  Let us denote
\be
\widehat{V}_{i}=\beta^{-\frac{4}{3}\widehat{L}_0}\widehat{L}_+^{(i)}\beta^{\frac{4}{3}\widehat{L}_0}=\sum_{k>0}a_{k}^{(i)} u^{k} \widehat{L}_{k}
\ee
for $i=1,2$. Then, the relation (\ref{FromKWtoH}) connecting two tau-functions can be represented as 
\be
\tau_{Hodge}({\bf t};u)=e^{-\widehat{V}_{1}}e^{-\widehat{V}_{2}}\tau_{KW}({\bf t^o}).
\ee

This relation can be simplified with the help of the Virasoro constraints for the KW tau-function (\ref{KWvir}). Indeed, from (\ref{hmin}) it follows that the series $f_{+2}(z)$ includes only odd powers of $z$, so that the operator $\widehat{L}^{(2)}_{+}$ and, as a consequence, $\widehat{V}_{2}$, contain only positive even Virasoro operators:
\be\label{opop}
\widehat{V}_{2}=\sum_{k=1}^\infty a_{2k}^{(2)} u^{2k} \widehat{L}_{2k}.
\ee 

Let us forget for a while that the coefficients $a_{2k}^{(2)}$ are defined by (\ref{hmin}), and consider the operator $\widehat{V}_{2}$ for arbitrary $a_{2k}^{(2)}$. 
Then, from the Baker--Campbell--Hausdorff formula it follows that for any such operator there exists a unique first order operator
\be
\widehat{N}=\sum_{k=1}^\infty b_{k} u^{{2k}}\frac{\p}{\p t_{2k+3}}
\ee
such that
\be\label{Simp}
\exp\left(\widehat{N}\right)\cdot\exp\left(-\widehat{V}_{2}\right) =\exp\left(-2\sum_{k=1}^\infty a_{2k}^{(2)} u^{{2k}} \widehat{L}_{k}^{KW} \right).
\ee
Virasoro operators $\widehat{L}_{k}^{KW}$ annihilate the KW tau-function, thus from (\ref{Simp}) it follows that 
\be
\exp\left(-\widehat{V}_{2}\right)\tau_{KW}({\bf t})=\exp\left(-\widehat{N}\right)\tau_{KW}({\bf t}).
\ee
To find $\widehat{N}$ we switch to the ordinary differential operators (\ref{obdif}).
Then identity (\ref{Simp}) is equivalent to
\be\label{Simp1}
\exp(n(z))\cdot\exp\left(-\sum_{k=1}^\infty a_{2k}^{(2)} l_{2k} \right)=\exp\left(-\sum_{k=1}^\infty a_{2k}^{(2)} \left(l_{2k} -z^{-2k-3}\right)\right),
\ee
where $n(z)$ is a formal Taylor series in $z^{-1}$ (it is obvious that there are no central terms which would modify the relation).
Lemma \ref{LEM1} allows us to restore $n(z)$. Indeed, if $g(z)=-\sum_{k=1}^\infty a_{2k}^{(2)}z^{1-2k}$, then (\ref{Simp1}) is equivalent to (\ref{lem2}) for $k=4$ and $\alpha=1$, that is to (\ref{k4case}):
\be
n(z)=\frac{1}{3 (f(z))^3}-\frac{1}{3z^3}.
\ee

Finally, we come back to the particular coefficients $a_{2k}^{(2)}$. If the coefficients $a_{2k}^{(2)}$ are as in Section \ref{Con}, then $f(z)$ coincides with $f^{-1}_{+2}(z)$, and
from (\ref{hmin}) we have 
\be
n(z)=\frac{1}{3 (f_{+2}^{-1}(z))^3}-\frac{1}{3z^3}=\frac{1}{z^2}\coth\left(\frac{1}{z}\right)-\frac{1}{z}-\frac{1}{3z^3}=\sum_{k=2}^\infty \frac{2^{2k}B_{2k}}{(2k)!}\frac{1}{z^{2k+1}} \\
=-\frac{1}{45}\,{z}^{-5}+{\frac {2}{945}}\,{z}^{-7}-{\frac {1}{4725}}\,{z}^{-9}
+{\frac {2}{93555}}\,{z}^{-11}-{\frac {1382}{638512875}}\,{z}^{-13}+O \left( {z}^{-15} \right) 
\ee
where $B_k$ are the Bernoulli numbers. Coming back to the operators from $W_{1+\infty}$ we have  
\be
\widehat{N}=\sum_{k=2}^\infty\frac{2^{2k}B_{2k}}{(2k)!}u^{{2}(k-1)} \frac{\p}{\p t_{2k+1}}. 
\ee
This operator shifts all KdV times.\footnote{An operator of this type in the context of the Kontsevich--Witten tau-function appears in the description of the Weil--Petersson volume of the moduli spaces of bordered Riemann
surfaces \cite{MIRZAKHANI}.}
Thus we have proved 
\begin{lemma}
The relation (\ref{FromKWtoH}) holds if and only if 
\be\label{HKW}
\tau_{Hodge}({\bf t},u)=e^{-\widehat{V}_{1}}e^{-\widehat{N}}\tau_{KW}({\bf t}).
\ee
\end{lemma}
In the next section we will show that this relation is satisfied, at least up to a constant factor.

\subsection{Proof of Theorem \ref{THEOR}}\label{PR1}
As it was shown in the previous section, the statement of Theorem \ref{THEOR} is equivalent to the relation
\be\label{HKW1}
\tau_{KW}({\bf t})=C(u)e^{\widehat{N}}e^{\widehat{V}_{1}}\tau_{Hodge}({\bf t},u).
\ee
Let us prove it using the description of the Grassmannian considered in Section 2. Namely, we will show that the vectors
\be
\Omega_k= u^{k-1}e^{\tilde{n}} e^{v_{1}} \Phi_k^{Hodge},\,\,\,\,\, k>0
\ee
belong to the KW space of the Grassmannian
\be\label{relphi}
\Omega_k\in {\mathcal W}_{KW}.
\ee
Here
\be
v_{1}=\sum_{k>0}a_{k}^{(1)} u^{{k}} {l}_{k}
\ee
and
\be
\tilde{n}=u^{-3}\, n(u^{-{1}}z).
\ee
are the counterparts of the operators $\widehat{V}_1$ and $\widehat{N}$.
According to Section \ref{MiwaGrass}, the relation (\ref{relphi}) guarantees that two tau-functions coincide up to a constant factor. 

For small $k$ the relation (\ref{relphi}) can be checked perturbatively. For example, 
\be
\Omega_1=
\left(1+\frac{4}{15z^2}u^2+\left(\frac{47}{1620}-\frac{2}{135z^6}\right)u^3+\left(\frac{1}{4050z^{10}}+\frac{1}{63z^4}\right)u^4+\cdots\right)\Phi_1^{KW}(z)\\+\left(\frac{2}{3}u-\frac{1}{45z^4}u^2+\frac{4}{45z^2} u^3-\left(\frac{23}{5670z^6}+\frac{673}{68040}\right)u^4+\cdots\right)\Phi_2^{KW}(z)\\ 
\Omega_2=\left(1+\frac{4}{3z^2}u^2+\left(\frac{1361}{1620}-\frac{1}{27z^6}\right)u^3+\left(\frac{1}{4050z^{10}}+\frac{1097}{4725z^4}\right)u^4+\cdots\right)\Phi_1^{KW}(z)\\+\left(\frac{5}{3}u-\frac{1}{45z^4}u^2+\frac{32}{45z^2} u^3+\left(-\frac{787}{28350z^6}+\frac{10877}{68040}\right)u^4+\cdots\right)\Phi_2^{KW}(z) 
\ee
and, since $b_{KW}=z^{-2}$ is the Kac--Schwarz operator for the KW tau-function, the relation (\ref{relphi}) for $k=1$ and $k=2$ holds at least up to $u^4$. The main idea of the proof is to show that all $\Omega_k$ can be represented in a similar way to all orders in $u$.

To prove (\ref{relphi}) we use Lemma \ref{SING}. It shows that the vector $u^{k-1}\Phi_k^{Hodge}$ can be considered as a formal series in variable $u$. Let us introduce a new variable $\eta=u/z$. Then, in the variables $u$ and $\eta$ the basis vector $\Phi_k^{Hodge}$ from (\ref{bvhodge}) is given by 
\be
u^{k-1}\Phi_k^{Hodge}=\frac{e^{G}}{\sqrt{2 \pi  u^3 \eta^{-1}}}\int_{-\infty}^\infty   d y\, \exp\left({-\frac{y^2}{2 u^3}+\left(\frac{1}{2}-k\right)(y-T)+\frac{e^{y-T}}{u^3}}\right),
\ee
where
\be
T=\frac{1}{1+\eta}+\log(1+\eta)
\ee
and
\be
G=-\frac{u^3}{2}\left(k-\frac{1}{2}\right)^2+\frac{1}{2u^3}\left(\frac{\eta^2}{(1+\eta)^2}-1\right)-\log(1+\eta).
\ee
Let us change the variable of integration $y \mapsto u\,y+T$:
\be
u^{k-1}\Phi_k^{Hodge}=\frac{e^{{G-\frac{T^2}{2 u^3}}}}{\sqrt{2 \pi u \eta^{-1} }}\int_{-\infty}^\infty   d y\, \exp\left({-\frac{y^2}{2 u}-y\,u\left(k-\frac{1}{2}+\frac{T}{u^3}\right)+\frac{e^{yu}}{u^3}}\right).
\ee
Action of the operator $\exp(v_{1})$ on the basis vectors is defined by (\ref{viract}).  It maps the variable $\eta$ to
\be
\eta \mapsto \frac{1}{f_{+1}(\eta^{-1})}=\frac{e^\eta \sinh(\eta)}{\eta}-1,
\ee
so that $T$ maps to
\be
\widetilde{T}=\exp \left({v_{1}}\right)\,T\,\exp\left({-v_{1}}\right)=\frac{2\eta}{e^{2\eta}-1}+\log \left(\frac{e^{2\eta}-1}{2\eta}\right)= 
1+\frac{1}{2}\eta^2+O(\eta^4).
\ee
Combining it with expression 
\be
\tilde{n}=\frac{1}{u^3}\left(\eta^2 \coth \eta-\eta-\frac{\eta^3}{3}\right)
\ee
 after some calculations we obtain 
\be\label{newint}
\Omega_k=\frac{e^{\widetilde{G}-\frac{\eta^3}{3u^3}}}{\sqrt{2 \pi u \eta^{-1} }}\int_{-\infty}^\infty   d y\, \exp\left(-\widetilde{W}\right),
\ee
where
\be
\widetilde{G}=-\frac{u^3}{2}\left(k-\frac{1}{2}\right)^2+\frac{1}{u^3}\left(\eta^2\frac{e^{4\eta}+1}{\left(e^{2\eta}-1\right)^2} -\eta\frac{e^{2\eta}+1}{e^{2\eta}-1}-\frac{\widetilde{T}^2}{2}\right)\\
+\log\left(\left(\frac{\eta^{-1}}{f_{+1}(\eta^{-1})}\right)^{\frac{3}{2}}\sqrt{f_{+1}'(\eta^{-1})}\,\frac{2\eta}{e^{2\eta}-1}\right)
\ee
and
\be
\widetilde{W}={-\frac{y^2}{2 u}-y\,u\left(k-\frac{1}{2}+\frac{\widetilde{T}}{u^3}\right)+\frac{e^{yu}}{u^3}}.
\ee
Since both $v_1$ and $\tilde{n}$ are series in $u$, the vectors (\ref{newint}) are still the formal Taylor series in the variable $\eta$ with coefficients depending on both positive and negative powers of $z$.

The series $\widetilde{T}$ is invariant under a sign change of $\eta$:
\be
\widetilde{T}\stackrel{\eta \to -\eta}{\longmapsto} - \frac{2\eta}{e^{-2\eta}-1}+\log \left(-\frac{e^{-2\eta}-1}{2\eta}\right)=\frac{2\eta\, e^{2\eta}}{e^{2\eta}-1}+\log \left(\frac{e^{2\eta}-1}{2\eta}\right)-2 \eta=\widetilde{T}.
\ee
Thus, it contains only even powers of $\eta$. It is easy to show that $\widetilde{G}$ is also invariant under the change $\eta\mapsto -\eta$ so that its expansion is
\be
\widetilde{G}=-\frac{u^3}{2}\left(k-\frac{1}{2}\right)^2-\frac{1}{u^3}\left(1+\frac{1}{24}\eta^4-\frac{13}{3240}\eta^6+\dots\right)-\frac{1}{18}\eta^2+\frac{7}{1620}\eta^4-\frac{137}{382725}\eta^6+\dots
\ee
Since
\be
\widetilde{W}=\frac{1}{u^3}+\frac{y^3}{3!}-\frac{y\eta^2}{2u^2}+\sum_{m=4}^\infty\frac{y^m}{m!}u^{m-3}-yu\left(k-\frac{1}{2}-\frac{1}{u^3}\left( \frac{1}{36}\eta^4-\frac{1}{405}\eta^6+\dots\right)\right),
\ee
also contains only even powers of $\eta$, coming back to variables $u$ and $z$, we have
\be
\widetilde{G}+\widetilde{W}=\frac{y^3}{3!}-\frac{y}{2z^2}+u \sum_{i,j=0}^\infty \kappa_{ij}(u)y^i z^{-2j},  
\ee
where $\kappa_{ij}(u)$ are some polynomials in $u$.\footnote{More specifically
\be
\kappa_{00}=-\frac{u^2}{2}\left(k-\frac{1}{2}\right)^2,\\
\kappa_{10}=-\left(k-\frac{1}{2}\right),\\
\kappa_{20}=\kappa_{30}=0,\\
\kappa_{i0}=\frac{u^{i-4}}{4!},\,\,\,\,\,\mathrm{for} \quad i>3,\\
\kappa_{i1}=-\frac{1}{18}u\,\delta_{i,0},\\
\kappa_{i2}=\left(\frac{7}{1620}u^3-\frac{1}{24}\right)\delta_{i,0}+\frac{u}{36}\delta_{i,1},
\ee
and, for general $j>2$ we have $\kappa_{ij}=\left(\left(a_ju^3+b_j\right)\delta_{i,0}+c_j\,u\,\delta_{i,1}\right) u^{2(j-2)}$ for some rational $a_j$, $b_j$, and $c_j$.
}

Thus,
\be
\Omega_k= \frac{e^{-\frac{1}{3z^3}}}{\sqrt{2 \pi z }}\int_{-\infty}^\infty   d y\, \left(\sum_{i,j=0}^\infty \tilde{\kappa}_{ij} (u) y^i  z^{-2j} \right)
\exp\left(\frac{y^3}{3!}-\frac{y}{2z^2}\right), 
\ee
where $\tilde{\kappa}_{ij} (u)$ are some other polynomials in $u$; or, using expression (\ref{KWbint}) of the basic vectors for the KW tau-function, we have
\be
\Omega_k= \sum_{i,j=0}^\infty \tilde{\kappa}_{ij} (u)  z^{-2j} (-1)^i \Phi^{KW}_i(z).
\ee
Because $z^{-2}$ is the Kac--Schwarz operator, this is a combination of the basis vectors for the KW tau-function, which completes the proof of Theorem \ref{THEOR}.


\subsection{Kac--Schwarz description of the Hodge tau-function}

The relation (\ref{HKW}) (or, equivalently, relation (\ref{HKW1})) allows us to obtain the Kac--Schwarz operators for the Hodge tau-function. Indeed, the Kac--Schwarz operators for two tau-functions are related by the conjugation:
\be\label{KSH3}
a_{Hodge}=e^{-v_{1}}\,e^{-\tilde{n}}\,a_{KW}\,e^{\tilde{n}}\,e^{v_{1}},\\
b_{Hodge}=e^{-v_{1}}e^{-\tilde{n}}\,b_{KW}\,e^{-\tilde{n}}e^{v_{1}}.
\ee
Since both $n(z)$ and $b_{KW}$ are series in $z$ with constant coefficients, they commute. Thus, we have
\be
b_{Hodge}=\left(\frac{q}{u}\right)^2,
\ee
where 
\be\label{Asq}
q\equiv u\,e^{-v_{1}}\,z^{-1}\,e^{v_{1}}=\widetilde{f}_{+1}(\eta)=\eta-\frac{2}{3}\eta^2+\cdots
\ee
with $\eta\equiv u/z$ and the series $\widetilde{f}_{+1}$ is related to $f_{+1}$ by (\ref{ftilde}). Thus, the series $q$ is defined by an equation, which follows from (\ref{negop}):
\be\label{qeq}
\eta=\frac{e^{q}\sinh\left(q\right)}{q}-1.
\ee
To solve it we use the ansatz 
\be
q=\frac{S_--\frac{1}{1+\eta}}{2}.
\ee
On substitution of this ansatz into (\ref{qeq}) we get an equation
\be
S_-e^{-S_-}=\frac{1}{1+\eta}e^{-\frac{1}{1+\eta}},
\ee
where $S_-$ is the solution, which corresponds to the asymptotic (\ref{Asq}) for small $\eta$.
Thus,
\be
q=\frac{S_--S_+}{2},
\ee
where $S_\pm$ are two solutions of the equation 
\be
Se^{-S}=\frac{1}{1+\eta}e^{-\frac{1}{1+\eta}}
\ee
in the vicinity of the the ramification point $\eta=0$, $S=1$. They are
\be
S_+=\frac{z}{u+z}=\sum_{k=0}^\infty (-1)^k \eta^{k},\\
S_{-}=1+\eta-\frac{1}{3}\,{\eta}^{2}+\frac{1}{9}\,{\eta}^{3}-{\frac {1}{135}}\,{\eta}^{4}-{
\frac {19}{405}}\,{\eta}^{5}+{\frac {43}{567}}\,{\eta}^{6}-{\frac {3827}{
42525}}\,{\eta}^{7}+O(\eta^{8}),
\ee
so that 
\be
q(z)=\eta-\frac{2}{3}\,{\eta}^{2}+\frac{5}{9}\,{\eta}^{3}-{\frac {68}{135}}\,{\eta}^{4}+{\frac {193}{
405}}\,{\eta}^{5}-{\frac {262}{567}}\,{\eta}^{6}+{\frac {19349}{42525}}\,{\eta}
^{7}+O(\eta^8).
\ee

Therefore, the Kac--Schwarz operator
\be
b_{Hodge}={z}^{-2}-\frac{4}{3}\,u {z}^{-3}+{\frac {14}{9}}\,u^2 {z}^{-4}-{\frac {236}{135}}\,u^3{
z}^{-5}+{\frac {29}{15}}\,u^4{z}^{-6}-{\frac {6008}{2835}}\,u^5 {z}^{-7}+O
 \left( {z}^{-8} \right)
\ee
define an infinite family of the Kac--Schwarz operators $j_k^{Hodge}= (b_{Hodge})^k$. This family corresponds to the family of operators $\widehat{J}_{k}^{Hodge}$:
\be
\widehat{J}_{k}^{Hodge}=\widehat{Y}_{j_k^{Hodge}}=e^{-\widehat{V}_{1}}e^{-\widehat{N}}\,\frac{\p}{\p t_{2k}}\,e^{\widehat{N}}e^{\widehat{V}_{1}}
=\mbox{res}_z \left(Q^{2k} \widehat{J}\left(z\right)\right)
\ee
where $Q=u^{-1}\widetilde{f}_{+1}(uz)$. They are
$u$-deformations of the operators $\widehat{J}_{k}^{KW}$, which define the KdV reduction of the KW tau-function:
\be
\widehat{J}_{k}^{Hodge}=\frac{\p}{\p t_{2k}}+O(u),
\ee
for example, the coefficients of the first operator from this family coincide with the coefficients of $b_{Hodge}$
\be
\widehat{J}_1^{Hodge}
=\frac{\p}{\p t_2}-\frac{4}{3}\,u \frac{\p}{\p t_3}+{\frac {14}{9}}\,u^2\frac{\p}{\p t_4}-{\frac {236}{135}}\,u^3\frac{\p}{\p t_5}+{\frac {29}{15}}\,u^4\frac{\p}{\p t_6}-{\frac {6008}{2835}}\,u^5 \frac{\p}{\p t_7}+O(u^6).
\ee
The operators $\widehat{J}_{k}^{Hodge}$, by construction, annihilate $\tau_{Hodge}$:
\be
\widehat{J}_{k}^{Hodge}\, \tau_{Hodge}=0, \,\,\,\,\,\,k>0.
\ee

Since $b^{Hodge}$ is not a polynomial, one can consider other combinations of $b_{Hodge}$ and its powers. For example, another combination will appear in case of conjugation described by the relation (\ref{FromKWtoH}) instead of (\ref{HKW}).
It seems that there is no canonical choice. This corresponds to the fact that according to (\ref{KSH3}), both $a_{Hodge}$ and $b_{Hodge}$ are ``gap-infinity''  operators in the notations of \cite{AMSvM}.

Let us describe the Kac--Schwarz operator $a_{Hodge}$. Namely, the conjugation with $e^{\tilde{n}}$ in (\ref{KSH3}) yields
\be
e^{-\tilde{n}}\,a_{KW}\,e^{\tilde{n}}=a_{KW}-z^{3}\left(\frac{\p}{\p z}\tilde{n}\right)
=a_{KW}+\frac{z}{u^2}\left(2\,\eta\coth\left(\eta\right)-\eta^{2}\coth\left(\eta\right)^2-1\right).
\ee
From the conjugation relation (\ref{virsm}) it follows that conjugation with $e^{v_1}$ yields 
\be
a_{Hodge}=g_{Hodge}\,\p_z-\frac{1}{2}g_{Hodge}'-z^{-1}g_{Hodge}-h_{Hodge},
\ee
where
\be
h_{Hodge}=\frac{2(1-S_-)(1-S_+)}{u(S_--S_+)}=\frac{1}{qu}\frac{u}{u+z}\left(\frac{u}{u+z}-2q\right)
\ee
and
\be
g_{Hodge}=-\frac{2z^2u}{q}\left(1+\frac{u}{z}\left(1-\frac{1}{2q}\right)\right).
\ee
By construction, the Kac--Schwarz operators $a_{Hodge}$ and $b_{Hodge}$ satisfy the same canonical commutation relation (\ref{Can}) as operators for the KW tau-function
\be\label{Can}
\left[a_{Hodge},b_{Hodge}\right]_-=2.
\ee

Explicit expressions for operators $l_k^{Hodge}=-\frac{1}{4}\left[(b_{Hodge})^{k+1},a_{Hodge}\right]_+$ allow us to construct the Virasoro constraints for the Hodge tau-function. Namely, the operators
\be
\widehat{L}_{k}^{Hodge}=\widehat{Y}_{l_k^{Hodge}}-\frac{u^2}{36}\delta_{k,-1}+\frac{1}{16}\delta_{k,0}
\ee
annihilate the Hodge tau-function:
\be
\widehat{L}_{k}^{Hodge}\tau_{Hodge}=0,~~~k\geq-1.
\ee
Here the operators $\widehat{Y}_{l_k^{Hodge}}$ can be represented as residues
\be
\widehat{Y}_{l_k^{Hodge}}=\frac{1}{2} \mbox{res}_{z=0}\left(Q^{2k-1} \left(A\,\widehat{W}^{(2)}(z)+B\,\widehat{J}(z)\right)\right),
\ee
where
\be
A=2+zu\left(2-\frac{1}{uQ}\right),\\
B=\left(\frac{z}{1+uz}\right)^2-2\,Q\frac{z}{1+uz}.
\ee

Operators $\widehat{L}_{k}^{Hodge}$ are deformations of the Virasoro operators for the KW tau-function
\be
\widehat{L}_{k}^{Hodge}=e^{-\widehat{V}_{1}}e^{-\widehat{N}}\,\widehat{L}_k^{KW}\,e^{\widehat{N}}e^{\widehat{V}_{1}}=\widehat{L}_{k}^{KW}+O(u),
\ee
in particular 
\be
2\widehat{L}_{-1}^{Hodge}=\widehat{L}_{-2}-\frac{\p}{\p t_1}+u\left(2\widehat{L}_{-1}-\frac{2}{3}\frac{\p}{\p t_2}\right)+u^2\left(\frac{8}{9}\widehat{L}_{0}-\frac{4}{9}\frac{\p}{\p t_3}-\frac{1}{18}\right)\\
-u^3\left(\frac{2}{27}\widehat{L}_{1}+\frac{38}{135}\frac{\p}{\p t_4}\right)+u^4\left(\frac{11}{405}\widehat{L}_{2}-\frac{64}{405}\frac{\p}{\p t_5}\right)+O(u^5),\\
2\widehat{L}_{0}^{Hodge}=\widehat{L}_{0}-\frac{\p}{\p t_3}+\frac{1}{8}+u\left(\frac{2}{3}\widehat{L}_{1}+2\frac{\p}{\p t_4}\right)-u^2\left(\frac{2}{9}\widehat{L}_{2}+\frac{26}{9}\frac{\p}{\p t_5}\right)\\
+u^3\left(\frac{14}{135}\widehat{L}_{3}+\frac{494}{135}\frac{\p}{\p t_5}\right)-u^4\left(\frac{22}{405}\widehat{L}_{4}+\frac{1751}{405}\frac{\p}{\p t_6}\right)+O(u^5).
\ee

By construction, the operators $\widehat{L}_k^{Hodge}$ and $\widehat{J}_k^{Hodge}$ belong to $W_{1+\infty}$ and satisfy the commutation relations of the Heisenberg--Virasoro algebra
\be
\left[\widehat{J}_{k}^{Hodge},\widehat{J}_{m}^{Hodge}\right]_-=0,\\
\left[\widehat{L}_{k}^{Hodge},\widehat{J}_{m}^{Hodge}\right]_-=-m\,\widehat{J}_{k+m}^{Hodge},\\
\left[\widehat{L}_{k}^{Hodge},\widehat{L}_{m}^{Hodge}\right]_-=(k-m)\widehat{L}_{k+m}^{Hodge}.
\ee

The Virasoro constraints for the Hodge tau-function should be equivalent to the polynomial recursion formula for the Hodge integrals, which can be considered as a Laplace transform of the cut-and-join equation \cite{Mul}.

There is another possibility to construct the Kac--Schwarz operators for the Hodge tau-function. Namely, they can be obtained by a conjugation of the operators $a_H$ and $b_H$, which were obtained in Section \ref{HURWITZ}. While operators $a_{Hodge}$ and $b_{Hodge}$ obtained here are of zeroth and first order, conjugation of $a_H$ and $b_H$ yields operators of infinite order (which, nevertheless belong to the same algebra of the Kac--Schwarz operators). 





\subsection{Quantum spectral curve}

In this section we discuss quantum (spectral) curves for the three tau-functions, considered in this paper. In the literature, there are different notations for quantum spectral curve, see e.g., \cite{SchwarzQC,AMMqc,Chervov,Eynqc,Gukov,Dijk,Sha}. Here we take the point of view similar to the one adopted in \cite{qc1,qc2}. If 
\be\label{SC}
A(x,y)=0,\,\,\,\, x,y\in {\mathbb C}\,\, (\mbox{or}\,\, {\mathbb C}^*)
\ee
is the spectral curve, which describes the partition function $Z({\bf t})$ of the model, then an operator such that:\footnote{We put $\hbar=1$.}
\be
{A}^*\left(z,\frac{\p}{\p z}\right) Z\left(\left[z\right]\right)=0
\ee
which is a deformation (quantization) of the spectral curve (\ref{SC}), defines quantum spectral curve. The partition function in the Miwa parametrization $t_k=\frac{1}{k} z^k$, which is annihilated by the operator ${A}^*$, is known as the principal specialization of this partition function. 

If the partition function is a tau-function of an integrable hierarchy, the principal specialization coincides with the first basis vector
\be
\tau\left(\left[z\right]\right)=\Phi_1(z).
\ee
Thus, quantum curves for tau-functions are closely related to the Kac--Schwarz operators. If the operator ${A}^*\left(z,\frac{\p}{\p z}\right)$ describing a quantum curve is the Kac--Schwarz operator, one can ``quantize'' it once again using (\ref{Ydif}) and the boson-fermion correspondence. As a result one obtains an operator $\widehat{D}\equiv\widehat{Y}_{A^*}$, for which the full tau-function is an eigenfunction:
\be\label{QQ}
\widehat{D} \,\tau({\bf t})= C\, \tau({\bf t})
\ee
for some eigenvalue $C$.

For the Hurwitz tau-function the quantum curve was constructed in \cite{qc1}. Namely, the quantum curve operator is a quantization of the Lambert curve, and coincides with the Kac--Schwarz operator $a_H$ (\ref{Huraop}): 
\be
A^*_H=a_H=z \exp\left(\beta z \frac{\p}{\p z}\right) - z \frac{\p}{\p z}
\ee
so that, as follows from (\ref{KSHuw}),
\be
A^*_H\,\Phi_1^H(z)=0.
\ee
Then, quantization (\ref{Ydif}) and the boson-fermion correspondence yields an operator $\widehat{D}_H=\widehat{Y}_{A^*_H}$ such that
\be
\widehat{D}_H \,\tau_H({\bf t})=0.
\ee
Here 
\be
\widehat{D}_H=t_1-\widehat{L}_0+\beta \widehat{L}_{-1}+\frac{\beta^2}{2}\widehat{M}_{-1}+\cdots.
\ee

Equation (\ref{KWqc}) defines the quantum curve equation for the KW tau-functions:
\be
A^*_{KW}=a_{KW}^2-b_{KW}=\frac{e^{-\frac{1}{3z^3}}}{\sqrt{2\pi z}} \left(\left(z^3\frac{\p}{\p z}\right)^2- z^{-2}\right)\frac{\sqrt{2\pi z}} {e^{-\frac{1}{3z^3}}},
\ee
where, after the change of variables $x=z^{-2}$,  the operator in the brackets is given by a quantization of the Airy spectral curve $4\p_x^2-x$. The corresponding operator 
\be
\widehat{D}_{KW}=\widehat{M}_{-4}-2\widehat{L}_{-1}+t_4
\ee
annihilates the KW tau-function
\be
\widehat{D}_{KW} \,\tau_{KW}({\bf t})=0.
\ee
This equation is a combination of the constraints (\ref{KWfirst}) and (\ref{KWvir}).

The action of $W_{1+\infty}$ operators with positive energy can be translated to the action of operators from $w_{1+\infty}$ on the quantum curves. Thus, we claim that the quantum curve for the Hodge tau-function can be obtained by a conjugation from the quantum curve for the Hurwitz tau-function:
\be
A^*_{Hodge}= \beta^{-\frac{4}{3}l_0}e^{v_-}e^{ -\sum_{k=1}^\infty\frac{k^{k-2}}{k!} \beta^{k-1} z^k}\,A^*_{H}\, e^{ \sum_{k=1}^\infty\frac{k^{k-2}}{k!} \beta^{k-1} z^k}e^{-v_-} \beta^{\frac{4}{3}l_0},
\ee
where
\be
v_-=\sum_{k=1}^\infty a_{-k} \beta^{k} l_{-k}.
\ee

\newpage
\section{Concluding remarks}\label{CONC}

In this paper we have proved a relation, which connects three tau-functions of enumerative geometry by operators from $GL(\infty)$. However, there are still many open questions. First of all, Conjecture \ref{CONJ} remains unproven. An approach we adopt in this paper seems to be not suitable to prove that the constant $C(u)$ from Theorem \ref{THEOR} is equal to one.

It is natural to consider some relatively small subgroups of the $GL(\infty)$ and to investigate the corresponding families of tau-functions. The simplest subgroup corresponds to the Heisenberg--Virasoro algebra. However, all tau-functions, that can be obtained by the action of the group elements of the Heisenberg--Virasoro group on the trivial tau-function are too simple to describe partition functions that are interesting for applications. Namely, they are of the form
$$
\tau({\bf t})=\exp \left(\sum_{i,j} A_{ij} t_i t_j+ \sum_{j} B_j t_j\right)
$$
for some constant $A_{ij}$ and $B_j$. We claim that to obtain a family of tau-functions that is important for applications, it is enough to consider the group elements, which include also an operator $\widehat{M}_k$ from the $W^{(3)}$ algebra. Examples of such tau-functions are given by the three tau-functions considered in this paper. The method used for construction of the normal matrix model, developed here, can be generalized to other tau-functions, which are described by the group elements of the Heisenberg-$W^{(3)}$ algebra. It would also be interesting to find a relation between the operator representation for the KW tau-function obtained here and the similar (but essentially different) representation, obtained in \cite{AKW}.

Another question, which remains beyond the scope of this work, is the relation between obtained Virasoro constraints and corresponding matrix models. Usually the Virasoro constraints are related to simple symmetries of matrix integrals \cite{MMVir}, but we are unable to derive the Virasoro constraints for the obtained matrix integrals in this simple way. 

It would be natural to try to continue the chain of conjugations (\ref{KSH3}) and to construct the Kac--Schwarz operators of the zeroth and first orders for the Hurwitz tau-function. Unfortunately, the naive conjugation seems to give divergent results. 

The developed approach can be applied to the so-called $r$-spin Hurwitz numbers. To describe them, one should consider the generalized Kontsevich models instead of the KW tau-function \cite{ToA}.

\section*{Acknowledgments}
This work was supported in part by RFBR grant 12-01-00482, by the ERC Starting Independent Researcher Grant StG No. 204757-TQFT and
by Federal Agency for Science and Innovations of Russian Federation. The author is grateful to all participants of the MPG ITEP seminar for useful comments.


\newpage

 \begin{thebibliography}{37}

  
\bibitem{Kazarian}  
M.~Kazarian,
``KP hierarchy for Hodge integrals,''
Adv.\ Math.\ {\bf 221} (2009) 1-21,
arXiv:0809.3263 [math.AG].

  \bibitem{MMKH}
 A.~Mironov and A.~Morozov,
  ``Virasoro constraints for Kontsevich-Hurwitz partition function,''
  JHEP {\bf 0902} (2009) 024,
  arXiv:0807.2843 [hep-th].
  
  
\bibitem{BM}
  V.~Bouchard and M.~Marino,
  ``Hurwitz numbers, matrix models and enumerative geometry,''
  Proc.\ Symp.\ Pure Math.\  {\bf 78} (2008) 263,
  arXiv:0709.1458 [math.AG].


\bibitem{AZ}
A.~Alexandrov and A.~Zabrodin,
  ``Free fermions and tau-functions,''
  J.\ Geom.\ Phys.\  {\bf 67} (2013) 37,
  arXiv:1212.6049 [math-ph].
  
 \bibitem{BBT} O. Babelon, D. Bernard and M. Talon,
{\it Introduction to classical integrable systems},
Cambridge University Press, 2003. 

\bibitem{MorUFN}
  A.~Morozov,
  ``Integrability and matrix models,''
  Phys.\ Usp.\  {\bf 37} (1994) 1,
  arXiv:hep-th/9303139.  
  
  
\bibitem{Kyoto} E. Date, M. Jimbo, M. Kashiwara and T. Miwa,
``Transformation groups for soliton equations,''
in {\it Nonlinear integrable systems -- classical and quantum'},
eds. M. Jimbo and T. Miwa, World Scientific, pp. 39-120 (1983);
 M. Jimbo and T. Miwa, ``Solitons and
infinite dimensional Lie algebras,'' Publ. RIMS, Kyoto Univ.
{\bf 19} (1983) 943-1001.

\bibitem{JMbook} T. Miwa, M. Jimbo and E. Date, {\it 
Solitons: Differential Equations, Symmetries and Infinite
Dimensional Algebras}, Cambridge University Press, 2000.

\bibitem{Fukuma}
  M.~Fukuma, H.~Kawai and R.~Nakayama,
  ``Infinite dimensional Grassmannian structure of two-dimensional quantum gravity,''
  Commun.\ Math.\ Phys.\  {\bf 143} (1992) 371.
  
\bibitem{Bombay}
  V.~G.~Kac and A.~K.~Raina,
  ``Bombay Lectures on Highest Weight Representations of Infinite Dimensionsal Lie Algebras,''
  Adv.\ Ser.\ Math.\ Phys.\  {\bf 2} (1987) 1.

\bibitem{OP}  
 V.~Kac and A.~Radul,
  ``Quasifinite highest weight modules over the Lie algebra of differential operators on the circle,''
  Commun.\ Math.\ Phys.\  {\bf 157} (1993) 429
  arXiv:hep-th/9308153.

  
\bibitem{W3}
 A.~B.~Zamolodchikov,
  ``Infinite Additional Symmetries in Two-Dimensional Conformal Quantum Field Theory,''
  Theor.\ Math.\ Phys.\  {\bf 65} (1985) 1205
   [Teor.\ Mat.\ Fiz.\  {\bf 65} (1985) 347].

  \bibitem{HHK}
I.~P.~Goulden and D.~M.~Jackson,
``Transitive factorisations into transpositions and holomorphic mappings on the sphere'',
Proc. A.M.S.,
{\bf125} (1997), 51--60.

\bibitem{HHK1}
R.~Vakil,
 ``Enumerative geometry of curves via degeneration methods,''
Harvard Ph.D. thesis, 1997.
  
  \bibitem{Sato}
M. Sato, ``Soliton equations as dynamical systems on inÞnite dimensional Grassmann manifolds," RIMS Kokyuroku {\bf439} (1981) 30Ð40.

  
\bibitem{Segal}
G.~Segal, G.~Wilson, ``Loop groups and equations of KdV type," Publications Math\'ematiques de l'IH\'ES {\bf 61} (1985): 5-65.  

\bibitem{Kac}
  V.~Kac and A.~S.~Schwarz,
  ``Geometric interpretation of the partition function of 2-D gravity,''
  Phys.\ Lett.\ B {\bf 257} (1991) 329. 
  
  \bibitem{AMSvM}
M.~Adler, A.~Morozov, T.~Shiota and P.~van Moerbeke,
  ``New matrix model solutions to the Kac-Schwarz problem,''
  Nucl.\ Phys.\ Proc.\ Suppl.\  {\bf 49} (1996) 201
  arXiv:hep-th/9603066.
  
\bibitem{PM}
  F.~J.~Plaza Martin,
  ``Algebro-geometric solutions of the string equation,''
  arXiv:1110.0729 [math.AG].   
  
  

\bibitem{Kharchev}
  S.~Kharchev,
  ``Kadomtsev-Petviashvili hierarchy and generalized Kontsevich model,''
  arXiv:hep-th/9810091.

\bibitem{Shiota}
T.~Shiota,
``Characterization of Jacobian varieties in terms of soliton equations,''
Invent.Math., {\bf 83}, (1986), 333-382.

\bibitem{EqH}
  S.~Kharchev, A.~Marshakov, A.~Mironov and A.~Morozov,
  ``Landau-Ginzburg topological theories in the framework of GKM and equivalent hierarchies,''
  Mod.\ Phys.\ Lett.\ A {\bf 8} (1993) 1047;
   [Theor.\ Math.\ Phys.\  {\bf 95} (1993) 571]
   [Teor.\ Mat.\ Fiz.\  {\bf 95} (1993) 280]
  arXiv:hep-th/9208046.
  

  
  
  \bibitem{ELSV} T.~Ekedahl, S.~Lando, M.~Shapiro and A.~Vainshtein,
``Hurwitz numbers and intersections on moduli spaces of
curves,'' Invent.\ Math.\  {\bf 146} (2001) 297-327. arXiv:math/0004096 [math.AG].
    
  \bibitem{Mul}  M.~Mulase and N.~Zhang, ``Polynomial recursion formula for linear Hodge integrals,'' Communications in Number Theory and Physics {\bf 4}, (2010), 267-294,
arXiv:0908.2267.

\bibitem{GJV}
I.~Goulden, D.~Jackson, R.~Vakil,
``The Gromov-Witten potential of a point, Hurwitz numbers, and Hodge integrals,''
Proc.\ London\ Math.\ Soc.\ {\bf 83} (2001) 563-581,
 [math/9910004 [math.AG]].
 

\bibitem{Buryak}
A.~Buryak, ``Dubrovin-Zhang hierarchy for the Hodge integrals" arXiv:1308.5716.

\bibitem{Mumford}  
D.~Mumford,
``Towards enumerative geometry on the moduli space of
curves'', In: {\sl Arithmetics and Geometry} (M. Artin, J. Tate eds.), {\bf v.2},
Birkhauser, 1983, 271-328.

  \bibitem{FP}
  C.~Faber, R.~Pandharipande,
  ``Hodge integrals and Gromov - Witten theory,''
  Invent.\ Math.\ {\bf 139} (2000) 173-199, arXiv:math/9810173[math.AG]. 


\bibitem{TodaOk}
A.~Okounkov, ``Toda equations for Hurwitz numbers," 
Math. Res. Lett. {\bf 7} (2000) 447-453,
arXiv:math/0004128[math.AG]. 


\bibitem{Fromto}
  A.~Alexandrov,
  ``From Hurwitz numbers to Kontsevich-Witten tau-function: a connection by Virasoro operators,''
  Lett.\ Math.\ Phys.\  {\bf 104} (2014) 75
  arXiv:1111.5349 [hep-th].

\bibitem{Konts}
  M.~Kontsevich,
  ``Intersection theory on the moduli space of curves and the matrix Airy function,''
  Commun.\ Math.\ Phys.\  {\bf 147} (1992) 1.  
  
\bibitem{Witten}
 E.~Witten,
  ``Two-dimensional gravity and intersection theory on moduli space,''
  Surveys Diff.\ Geom.\  {\bf 1} (1991) 243.
  
 

  
  \bibitem{KS2}
  A.~S.~Schwarz,
  ``On some mathematical problems of 2-D gravity and W(h) gravity,''
  Mod.\ Phys.\ Lett.\ A {\bf 6} (1991) 611.
  
    
\bibitem{Dijk}
  R.~Dijkgraaf, L.~Hollands and P.~Sulkowski,
  ``Quantum Curves and D-Modules,''
  JHEP {\bf 0911} (2009) 047
  arXiv:0810.4157 [hep-th].  

  
  \bibitem{MMRP}
 A.~Alexandrov,
  ``Matrix Models for Random Partitions,''
  Nucl.\ Phys.\  {\bf B851 } (2011)  620-650.
  arXiv:1005.5715 [hep-th].

\bibitem{Morsh}
  A.~Morozov and S.~Shakirov,
  ``Generation of Matrix Models by W-operators,''
  JHEP {\bf 0904} (2009) 064
  arXiv:0902.2627 [hep-th].

\bibitem{MIRZAKHANI}
M. Mulase and B.Safnuk,  ``Mirzakhani's recursion relations, Virasoro constraints and the KdV hierarchy," math/0601194 (2006).



\bibitem{SchwarzQC}
  A.~Schwarz,
  ``Quantum curves,''
  arXiv:1401.1574 [math-ph].
  
\bibitem{AMMqc}
  A.~S.~Alexandrov, A.~Mironov and A.~Morozov,
  ``Unified description of correlators in non-Gaussian phases of Hermitean matrix model,''
  Int.\ J.\ Mod.\ Phys.\ A {\bf 21} (2006) 2481,
  arXiv:hep-th/0412099.
  
\bibitem{Chervov}
 A.~Chervov and D.~Talalaev,
  ``Quantum spectral curves, quantum integrable systems and the geometric Langlands correspondence,''
  arXiv:hep-th/0604128.
  
\bibitem{Eynqc}
  L.~Chekhov, B.~Eynard and O.~Marchal,
  ``Topological expansion of the Bethe ansatz, and quantum algebraic geometry,''
  arXiv:0911.1664 [math-ph].
 
\bibitem{Gukov}
 S.~Gukov and P.~Sulkowski,
  ``A-polynomial, B-model, and Quantization,''
  JHEP {\bf 1202} (2012) 070
  arXiv:1108.0002 [hep-th].  

\bibitem{Sha}
 M.~Mulase, S.~Shadrin and L.~Spitz,
  ``The spectral curve and the Schroedinger equation of double Hurwitz numbers and higher spin structures,''
  arXiv:1301.5580 [math.AG].

 \bibitem{qc1} 
  J.~Zhou,
  ``Quantum Mirror Curves for ${\mathbb C}^3$ and the Resolved Confiold,''
  arXiv:1207.0598 [math.AG].
  
 \bibitem{qc2}
   M.~Mulase and P.~Sulkowski,
  ``Spectral curves and the Schroedinger equations for the Eynard-Orantin recursion,''
  arXiv:1210.3006 [math-ph]. 
  
  
 \bibitem{AKW}
  A.~Alexandrov,
  ``Cut-and-Join operator representation for Kontsewich-Witten tau-function,''
  Mod.\ Phys.\ Lett.\ A {\bf 26} (2011) 2193,
  arXiv:1009.4887 [hep-th].
  
\bibitem{MMVir} 
  A.~Mironov and A.~Morozov,
  ``On the origin of Virasoro constraints in matrix models: Lagrangian approach,''
  Phys.\ Lett.\ B {\bf 252} (1990) 47.
  


\bibitem{ToA}
A.~Alexandrov, to appear.   




\end{thebibliography}
\end{document}